%% file: Arxiv_version.tex
\documentclass[onecolumn,11pt]{IEEEtran}
\input{frontmatter}
\usepackage{amsmath}%[cmex10]
\usepackage{amsthm}
\newtheorem{theorem}{Theorem}[section]
\newtheorem{lemma}[theorem]{Lemma}

\newtheorem{definition}[theorem]{Definition}
\newtheorem{example}[theorem]{Example}
\usepackage{graphicx}
\usepackage{subfigure}
\usepackage{float}
\usepackage{caption}
\usepackage{indentfirst}
\usepackage{setspace}
\usepackage{amssymb,bm}
\usepackage{algorithmic, algorithm}
\usepackage{leftidx}
\usepackage{array,color}
\usepackage{capt-of}
\usepackage{cite}

\begin{document}

\title{Bounds on Multiple Sensor Fusion}

\author{\IEEEauthorblockN{Bill Moran\IEEEauthorrefmark{1},
Fred Cohen\IEEEauthorrefmark{2},
Zengfu Wang\IEEEauthorrefmark{3},
Sofia Suvorova\IEEEauthorrefmark{4},
Douglas Cochran\IEEEauthorrefmark{5},
Tom Taylor\IEEEauthorrefmark{5},
Peter Farrell\IEEEauthorrefmark{4} and
Stephen Howard\IEEEauthorrefmark{6}} \\
\IEEEauthorblockA{\IEEEauthorrefmark{1}RMIT University, {bill.moran@rmit.edu.au}} \\
\IEEEauthorblockA{\IEEEauthorrefmark{2}University of Rochester, {fred.cohen@rochester.edu}}\\
\IEEEauthorblockA{\IEEEauthorrefmark{3}Northwestern Polytechnical University, {wangzengfu@gmail.com}} \\
\IEEEauthorblockA{\IEEEauthorrefmark{4}University of Melbourne, {ssuv, pfarrell@unimelb.edu.au}}\\
\IEEEauthorblockA{\IEEEauthorrefmark{5}Arizona State University, {cochran, tom.taylor@asu.edu}} \\
\IEEEauthorblockA{\IEEEauthorrefmark{6}Defence Science and Technology Organisation, {Stephen.Howard@dsto.defence.gov.au}}
}


\maketitle

%\tableofcontents

\begin{abstract}
We consider the problem of fusing measurements from multiple
  sensors, where the sensing regions overlap and data are non-negative
  --- possibly resulting from a count of indistinguishable discrete
  entities. Because of overlaps, it is, in general, impossible to fuse
  this information to arrive at an accurate estimate of the overall
  amount or count of material present in the union of the sensing
  regions. Here we study the range of overall values consistent with the
  data. Posed as a linear programming problem, this leads to
  interesting questions associated with the geometry of the sensor
  regions, specifically, the arrangement of their non-empty
  intersections. We define a computational tool called the fusion
  polytope and derive a condition for this to be in the positive
  orthant thus simplifying calculations. We show that, in two
  dimensions, inflated tiling schemes based on rectangular regions
  fail to satisfy this condition, whereas inflated tiling schemes
  based on hexagons do.
\end{abstract}

%\begin{keywords}
%Sensor Network, Data Fusion, Linear Programming, Polytope, Genericity, Extreme Points, Simplicial Complex
%\end{keywords}

\IEEEpeerreviewmaketitle

\section{Introduction}
Examples abound in sensing of measurement processes which, rather than
identifying objects or events, merely count them or measure their
size, for instance, by integrating a total response over all objects
accessible to each individual sensor. Such examples include
measurements of radioactivity using Geiger counters, people counting
algorithms in video-analytics that rely on some overall size of a
moving group rather than separate identification of each individual,
cell counting techniques, and counts of numbers of RF transmitters
using overall signal strength. In all of these circumstances,
measurement relies on a particular property of the object(s) being
measured.

At an abstract level, envisaged is a situation involving multiple
sensors each able to measure a different range of properties (such a
property might be an amount of some substance in a given spatial
region), and where the ranges of properties involved are not mutually
exclusive. Our initial interest was in the counting of spatially
distributed targets, and in this case the property is that of being in
a given ``sensor region'' described in terms of its geographical
spread.  The methods discussed here apply equally well to measurements
where the outcome is a real number, provided only that the quantities
being measured are non-negative, and to where the distinguishing
properties of the various sensors might be characteristics other than
physical location.  In fact, of course, it is enough that the
measurements have a known lower bound which in itself might be
negative, since they can be additively adjusted to provide
non-negative measurements in an obvious way.

As the results and ideas of this paper then, are very generic, our
methods will be clearer if we fix on the simple and, in some ways,
archetypal example that first motivated our interest.
This involves sensors on the ground capable of counting all objects
(``targets'') close to them in some sense. Each sensor is associated
with a ``sensor region'' within which any target present is
detected, without being identified, and forms part of the count for
that sensor. As already indicated, the particular
property of the objects is that they belong to this sensor region.
In this case,
the regions may be regarded, for simplicity, as subsets of
$\bR^{2}$. A more complex example might define the sensor region as
encompassing all transmissions that are both close to a sensor in
$\bR^2$ and emit in a certain frequency band; the regions in this case
are subsets of $\bR^{3}$. In more complex situations, targets may be
distinguished by their positions in space and by a number of other
features such as colour, emission frequency or energy, or rapidity of
movement~(in the case of radars measuring Doppler, for example). The
targets, then, can be regarded as points in a multi-dimensional space
$\bR^{n}$ and each sensor as defining a region of that space over
which it is able to detect and count targets or measure the total
integrated value of some response over that region.

An unrealistic aim would be to determine the total number of targets,
or to find the integral of all of the measured data, in the union of
all sensor regions. This is, of course, impossible because the regions
may overlap and we are not given information about the number of
targets/integrated measured values in the intersections of these
sensor regions. The focus question of this paper is merely to find the range
of possible values for the number of targets.

Another application area for the ideas and results presented here is
the assessment of the probability that at least one sensor of several
will see a given event. This kind of analysis is required if, for
instance, we are interested in obtaining a measure of performance for
the entire network of sensors: ``What is the probability that at least one sensor
will detect a target in the observed region?'' For such a question, each individual
sensor has a known probability of observing
the event $T$, say, and there are possibly unknown probabilities of
combinations of multiple sensors observing the event. Problems of this
kind are considered in \cite{Newell:2009:SCS:1817271.1817297}. An
interesting variant on this problem is explored in depth in
\cite{DBLP:conf/eccv/MittalD04,Lazos:2006:SCH:1167935.1167937,Wu:2007:EDS:1293372.1293563,Lazos:2007:PDM:1236360.1236426} where
detection of targets moving through an area~(two or three dimensional)~is studied.
In \loccit, sensors, each with
their own sensing region, are spread across the area of interest .
Targets move through the area in a linear motion and are detected with
a designated probability as they cross the region of a given sensor.
Since they cross multiple regions they may be detected more than once.
The overall probability of detection is required.

In all of those papers,  a version of the
inclusion-exclusion principle is employed to calculate an overall
probability of detection. This approach requires that some information
about the probabilities of multiple detections is known. For such
problems the results of this paper are able to provide a minimum
detectability performance consistent with the individual sensor
detection capability, with only minimal information about the joint
detection capabilities of multiple sensors. Specifically, this paper
will consider this minimum when we know which collections of detectors
are \emph{disjoint} that is, are incapable of detecting the same
event. This is a much less demanding requirement than that we know the
probabilities of multiple detections as required in the cited papers.

Connolly~\cite{Connolly85}, and
Liang~\cite{springerlink:10.1007/978-0-387-68372-0_6}  exploit the inclusion-exclusion principle to calculate the volumes of protein
molecules using NMR techniques. There again the ideas of this paper
might be used to provide a cruder assessment of the volume while
significantly reducing the number of measurements. Indeed our
techniques apply wherever the inclusion-exclusion principle could be
used if information about the counts for all intersections of regions  were available,
but instead only the information about the geometry of the  sensor regions is available.

To formulate the archetypal problem mathematically, we envisage a collection of
points (``targets'') $\bt_{1},\bt_{2},\bt_{3},\ldots,\bt_{m}$ in
$\bR^{n}$ and a collection of ``sensor regions''
$S_{1},S_{2},S_{3},\ldots,S_{R}$. At this stage, we impose no
structure on the sensor regions, other than that they are subsets of
$\bR^{n}$, perhaps with the additional proviso that they be Borel
measurable.

We emphasize at this point that,
while formulated in terms of counting targets, the same ideas apply to
all of the problems mentioned above, including the important case of
assessing probability of detection. Each target is assumed to belong
to the union $\bigcup_{r=1}^{R}S_{r}$ of the sensor regions. Whether
or not the sensors cover the region of interest is an issue not
discussed in this paper and we always assume that the region of
interest is covered by the sensors. Coverage problems of this type
have been considered by
Ghrist~\cite{ghrist08:_target_enumer_integ_planar_sensor_networ}. We finesse
this issue by always assuming that the target space is covered, or
rather that we are only interested in the region of observation
covered by at least one sensor. Of concern to us is that the sensor
regions may overlap so that a given target might be counted several times.
This problem, or variants of it, is discussed in \cite{Guo08,Huang07,Fang02,Aeron06,Prasad91} and many other papers.

Assumed known is the geometry of this situation; specifically, the
overlap regions $S_{r_{1}}\cap S_{r_{2}}\cap\cdots\cap S_{r_{T}}$ for
any set of $T$ distinct integers in the range $1\leq r_{t}\leq R$ is
known to be empty or non-empty.  It needs to be stressed that no
further information about the overlaps is known; in particular, it is
not known how many targets are in these overlaps. To be slightly more
specific, it is assumed that we know whether or not an overlap is
capable of containing a target of interest. This has to be specified a
little more precisely for some of our later results.

The structure of intersections can often be modelled by a
\emph{simplicial complex}, which approach we will discuss. This gives
rise to combinatorial problems, some of which have been addressed (see
Gr\"otschel and Lov\'{a}sz \cite{grot95} and Schrijver
\cite{Schrijver95}), at least in special cases. This setting also
admits direct application to a formula for the most likely number
which is analogous to a function defined on a simplicial complex. That
problem will be discussed in a sequel to this paper.

Each sensor reports its  measurement  to a central processor; thus
sensor $S_{r}$ reports $n_{r}$ ``targets''. The question at hand is ``how
many targets are there altogether?''; that is, how can we calculate
the overall  value  from these sensor reports? Of course, as the authors
of \cite{1505176} note (Theorem 1), it is trivial to see that there is
no unique answer to that question since we do not know how many
targets are in overlaps. The simple example of two overlapping sensor
regions with say $5$ targets reported by sensor $S_{1}$ and $7$
targets reported by sensor $S_{2}$ may have an overall  target count of
any number between $7$ and $12$ targets according to how many are in
the overlap region $S_{1}\cap S_{2}$. In more generality, the
inclusion-exclusion principle provides an answer to the overall  count
\emph{provided we know how many targets are in intersections of sensor regions}:
\begin{equation}
  \label{eq:inclusion-exclusion}
  |\bigcup_{r=1}^{R}S_r| =\sum_{r=1}^{R} |S_{r}|-\sum_{r_{1}\neq
    r_{2}} |S_{r_{1}}\cap S_{r_{2}}| + \sum_{r_{1}\neq
    r_{2}\neq r_{3}}|S_{r_{1}}\cap S_{r_{2}}\cap S_{r_{3}}|
  -\cdots+(-1)^{R-1}  |S_{1}\cap S_{2}\cap \cdots S_{R}|.
\end{equation}
Since the information about the number of targets in intersections is
typically unavailable, we cannot use this formula to calculate the overall number of
targets.

One might argue then that sensors should be chosen so that these sensor
regions do not overlap. As argued in \cite{1505176}, this is impractical for
several reasons. Sensor regions do not come in shapes that permit tiling of
Euclidean space or even the subset where targets might reside. Since we wish
to count all targets,  every  target should be in at least one sensor
region. In the situation where targets might be missed by an individual
sensor, it makes sense to have more and larger overlap rather than less.

The problem to be faced is to find the  range of possible values for the
number of targets from the information available: the sensor geometry
and the target count reports from individual sensors. This paper
addresses the problem of ascertaining the minimum number of targets
consistent with the target count reports (the maximum number is easily
calculated under our assumptions).

We will provide a description of this ``range of values'' problem in terms
of simplicial complexes and of linear programming. Reframing the
latter as a dual rather than primal domain enables us to describe, as
a computational tool, a polytope (the ``fusion polytope'') that is
dependent only on the geometry of the sensor region and independent of
the number of counts. In general the fusion polytope does not lie in
the positive orthant, but when it does lower bounds for the fused
information become much easier. Necessary and sufficient conditions
for the fusion polytope to lie in the positive orthant in terms of the
geometry of the sensor regions are given. For planar regions, we will provide a
description of some interesting sensor configurations that correspond to positivity
of the fusion polytope.

The theory will be illustrated with examples of simple cases in which
a description of the extreme points of the fusion polytope is
possible. The sensor configurations for which the simplicial complex
is a graph are discussed in some detail.
% Head 1
\section{Problem Formulation}\label{sec:2}
% Head 2
Our aim is to discuss the problem of counting of targets   by multiple
sensors. We recall that this is a surrogate for a large collection of
problems where target response is integrated over a sensor region, or
where the target might be the amount of some material being sensed and
so be non-negative real valued rather than non-negative integer
valued. Assume a collection of $R$ sensors, a \emph{sensor
  configuration} $\bS$, labelled by the regions they observe
$S_{1},S_{2},\ldots,S_{R}$, all subsets of $\bR^{n}$. While this
simple definition will suffice for much of this paper, later we will
need to be a little more stringent.

We assume the following properties of the collection of sensor regions:
\begin{description}
\item [\textit{Coverage}] \quad \quad That the union of the sensor regions
  $\bigcup_{r=1}^{R} S_{r}$ is the entire region of interest $\Omega$. Nothing
  escapes detection.
\item [\textit{Irredundancy}] \quad \quad \quad That there is no redundancy of sensors; that is, no
  sensor region is entirely contained in the union of the others.
\end{description}
While coverage is a fairly natural assumption; after all we are surely
only interested in the region that can be sensed, irredundancy is less
clear. In fact, it may well be unacceptable in some
applications. Nonetheless, it simplifies calculations and is not too
unreasonable.

At this stage, we remark that there are currently no topological
assumptions such as openness, closedness, or connectedness,  on the
sensor regions; they are merely sets with all of the potential
pathology that that entails. Later in the paper we shall impose some
topological restrictions which lead to interesting consequences.

A count made by each sensor is specified in a
\emph{sensor measurement  vector} $\bn=(n_1,n_2,\ldots, n_R)$.  In
general the measurements  are able to be non-negative real numbers without
changing the theory.
\begin{definition}
  \label{def:atom}
An \emph{atom} is a non-empty set of the form
\begin{equation}
  \label{eq:2aaa}
\langle i_{1},i_{2},\ldots,i_{T}\rangle= S_{i_{1}}\cap S_{i_{2}}\cap\cdots\cap S_{i_{T}}\cap S_{j_{1}}^{c}\cap
S_{j_{2}}^{c}\cap\cdots \cap S_{j_{S}}^{c},
\end{equation}
where ${}^{c}$ denotes set complement (with respect to the region of
interest $\Omega$),   and where $\{i_{1},i_{2},\ldots,i_{T},j_{1},j_{2},\ldots,j_{S}\}$
is an enumeration of the integers from $1$ to $R$ (without repetition of
course). The set of all atoms is denoted by $\fA(\bS)$.
\end{definition}
We recall that a Boolean algebra of sets is a collection of sets
closed under finite unions, intersections, and complements. Atoms are
minimal \emph{non-empty} elements of the Boolean algebra generated by
the sensor regions $S_{1},S_{2},\ldots, S_{R}$.  Note that, for any
atom, the number $T$ in \eqref{eq:2aaa} has to be positive, since the
intersection of all $S_{m}^{c}$ is empty because of the coverage
assumption. The whole space is the union of atoms, and these are
disjoint. The \emph{specification} of a sensor configuration $\bS$ is
really a statement of which intersections of the form \eqref{eq:2aaa}
are non-empty and which sensor regions these non-empty intersections
are contained in.

There is one further constraint that will require consideration.
\begin{definition} Given a sensor configuration $\bS=(S_1,S_2,\ldots,S_R)$,
if no non-empty intersection of sensor regions $\langle
i_{1},i_{2},\ldots,i_{R}\rangle$ is entirely contained in the union of different
sensor regions $\cup S_j$ where $j\not \in \{
i_{1},i_{2},\ldots,i_{N}\}$ then  the
sensor configuration is said to be \emph{generic}; otherwise the configuration is
\emph{degenerate}.
\end{definition}

The problem of finding the minimal overall values  differs significantly in the
degenerate case from the generic case.  Indeed the linear programming
formulation is considerably simpler for the generic
case. Unfortunately, there are many reasonable situations that are not
generic.  Examples (the simplest irredundant and a slightly more
complicated one) where the generic condition fails are given in Fig. \ref{fig:degeneracy},
but as we shall describe later, much more natural sensor
configurations can fail to be generic.
\begin{figure}[h!]
\centering
\subfigure[Simplest Degenerate Case: $S_{1}\cap S_{3}\subset
S_{2}$]{\includegraphics[width=0.35 \textwidth]{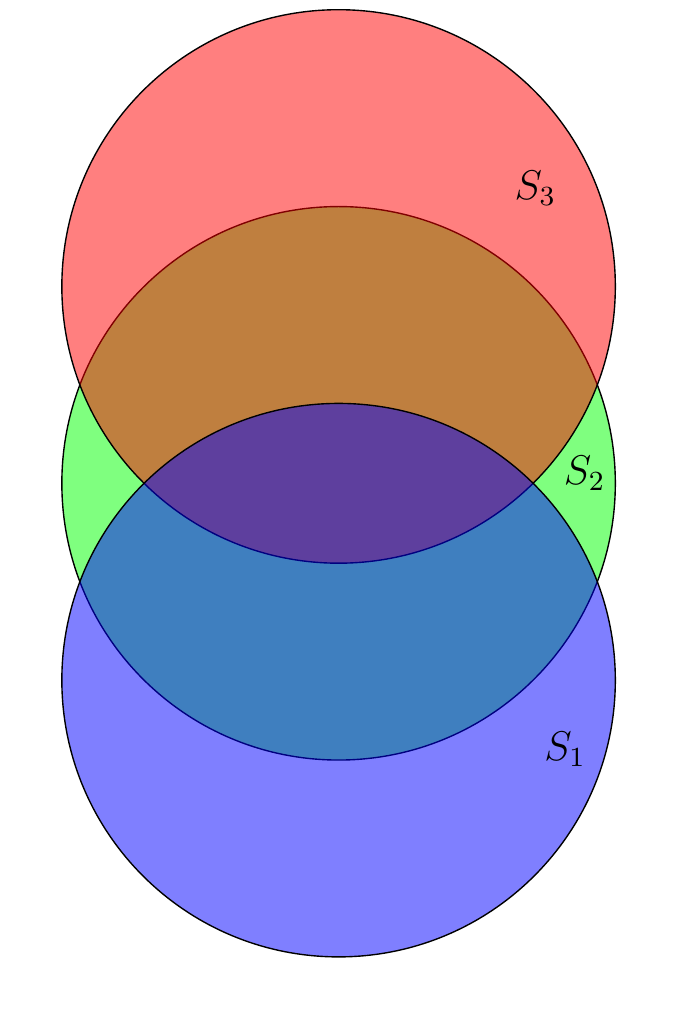}}
\subfigure[More Complicated Example: $S_{3}\cap S_{4}\cap S_{1}\subset S_{2}$]{ \includegraphics[width=0.4\textwidth]{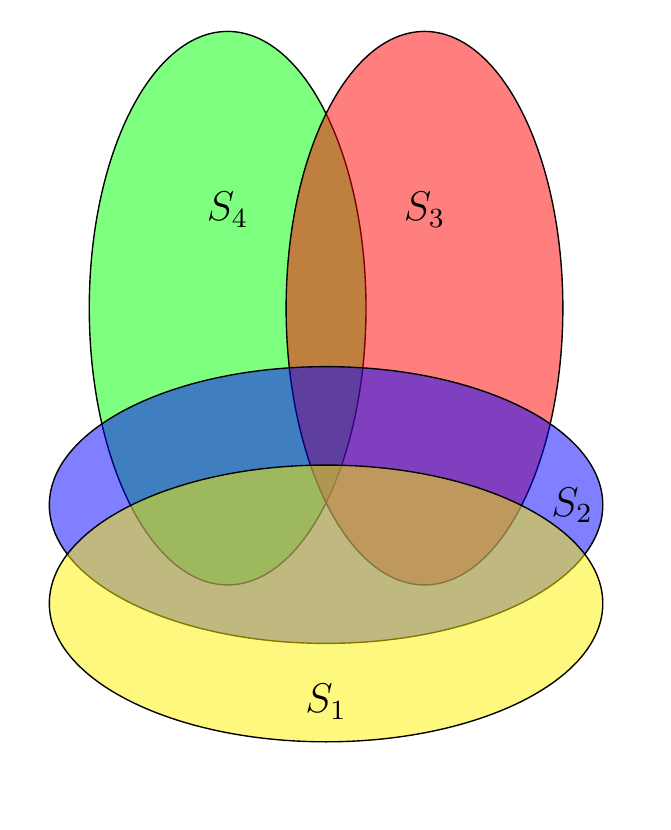}}
\caption{Simple cases of degeneracy}
  \label{fig:degeneracy}
\end{figure}

To be clear, we shall always assume coverage and irredundancy,
assumption of genericity is always stated explicitly.  Having defined
the basic structures, we reiterate the problem: given a sensor
configuration $\bS$ and a sensor measurement  vector $\bn$ our aim is to
investigate the problems of finding the range of possible overall
values consistent with $\bS$ and $\bn$.

\subsection{Simplicial Complex Formulation}

It is possible to provide a description of a sensor configuration in
terms of a geometrical object called a \emph{simplicial complex}. In
this context, a simplicial complex is defined to be a collection of
subsets $\Sigma$ of the set $\{1,2, \ldots,R\}$, where $R$ is the
number of sensors, with the property that if $\sigma\in \Sigma$ and
$\tau\subset \sigma$ then $\tau\in \Sigma$. All singletons $\{r\}$
$(r=1,2,\ldots,R)$ are assumed to be in $\Sigma$, as is the empty set
$\emptyset$. The \emph{dimension} of a \emph{simplex} $\sigma\in
\Sigma$ is just $\# \sigma -1$, where $\# A$ is the number of
elements of the set $A$. It is useful to think of points of
$\{1,2, \ldots,R\}$ (that is singletons) as \emph{vertices}, simplices
$\sigma$ of dimension $1$ as edges joining the elements of the set
$\sigma$, simplices $\sigma$ of dimension $2$ as triangles with
vertices the elements of $\sigma$, and so on. The resulting
geometrical object is called a \emph{geometrical realization} of the
simplicial complex.

The \emph{dimension} of a simplicial complex is the
maximum of the dimensions of all of its simplices.  The subsets of a simplex
$\sigma$ of dimension $\dim \sigma -1$ are called the \emph{faces}
of $\sigma$.

This abstract simplicial complex always has a geometrical
representation. The \emph{Geometric Realization
  Theorem} \cite{alexandroff35:_topol} states that a simplicial complex of
dimension $d$ has a geometric representation in $\bR^{2d+1}$, where
abstract simplices are represented by geometrical ones and where the
abstract concept of face corresponds to the geometrical faces of a
simplex. The geometrical picture is a very useful device in
visualising the structure of the problem we have described in
Section~\ref{sec:2},  and solutions to some special cases.

A given sensor configuration  $\bS=\{S_1,S_2,
\ldots, S_R\}$ then maps to a simplicial complex, called the
\emph{nerve} of the configuration as follows. The \emph{nerve} of
$\bS$ is the simplicial complex
$\Sigma(\bS)$ whose vertices are the numbers $\{1,2,\ldots, R\}$ and
where $\sigma\in \Sigma$ if $\bigcap_{r\in \sigma}S_r\neq\emptyset$.
It is straightforward to check that this is indeed a simplicial
complex. This definition is well-known; see for example
\cite{alexandroff35:_topol}.

The ability of the simplicial complex to represent the
intersection structure faithfully breaks down when the sensor
configuration is degenerate, as in Fig.~\ref{fig:degeneracy}. In the case of
Fig.~\ref{fig:degeneracy}(a), the nerve would consist of all
subsimplices of a triangle, but as can be seen there is no
distinguished region determined by intersections and differences of
sensor  regions that represents one of the $1$-simplices.  In effect $S_1\cap
S_2\cap S_3= S_1\cap S_2$ so that in the simplicial complex picture
two different simplices correspond to the same atom.

This simple example is characteristic, as the following simple theorem
shows.
\begin{theorem}
  \label{sec:5}
  Given a sensor configuration $\bS$,
  the correspondence which assigns an atom $\langle
  i_1,i_2,\ldots,i_t\rangle$  to the simplex $\sigma = \{i_1,
\ldots, i_t\}$ is $1-1$ from the atoms $\fA(\bS)$ of $\bS$ to (the
  simplices of) $\Sigma(\bS)$ if and only if the sensor configuration
  is generic.
\end{theorem}

\begin{proof}
  The assignment of
  \begin{enumerate}
  \item  a non-empty atom given by $S_{i_1}\cap \ldots \cap S_{i_t}$
    to the simplex $\sigma = \{i_1,\ldots, i_t\}$ and
  \item the empty intersection to $\emptyset$
  \end{enumerate}
  is a bijection of sets as long as the intersections $S_{i_1}\cap \ldots
  \cap S_{i_t}$ are non-empty and distinct with the single exception
  of the empty set. This  is exactly the definition of
    generic.
%    \mnote{this isn't the definition of generic we have ---
%    needs fixing}
\end{proof}

%*********************************
%*********************************
\section{Calculation of the Range of Overall Values}

We suppose a sensor configuration $\bS=\{S_1,S_2,\ldots, S_R\}$ and
sensor measurements  $n_1,n_2,\ldots, n_R$. For each atom $\sigma=\langle
i_{1}, i_{2}, \ldots, i_{r}\rangle$, assume that the amount in that
intersection is $m_{\sigma}$. The aim is to calculate all possible
overall  amounts in the entire region, consistent with the reports from
each sensor region. Since atoms are disjoint, this is
$\sum_{\sigma\text{ atom}}m_\sigma$.  An upper bound is the sum of the
sensor measurements
$\sum_{r=1}^R n_r$. It is essentially trivial to see that, if the
collection of sensor regions is irredundant, this upper bound is
indeed the maximum possible. It is achieved by putting all of the
value in a sensor region $S_r$ into $S_r\backslash \bigcup_{r'\neq r}
S_{r'}$.

Before going any further, we formulate the linear programming version
of the problem. At this point, the formulation does not require the
generic hypothesis, though this will be important later. The
variables in the \emph{primal} linear programming problem are the
values in each atom. The constraints in this case are that the values
in the atoms that are subsets of any given sensor region have to sum
to the sensor measurement  for that sensor region, since any such  region is a disjoint
union of its atoms. The problem is then to minimize the sum of all of
the atom values $m_\sigma$. Formally, the constraints are
\begin{equation}
  \label{eq:lin_prog_primal}
  \sum_{\stackrel{a\subset S_r}{ a \text{ is an atom}}} m_a =n_r \ (r=1,2,\ldots, R)
\end{equation}
and of course that $m_a\geq 0$ for all atoms $a$. To calculate the minimum
atom sum consistent with the sensor measurements, it is necessary to find the minimum
value of $\sum_{ a \text{  atom}} m_a$ subject to the constraints in
\eqref{eq:lin_prog_primal}. More succinctly, we let $A=A_{\cS}$, with
the subscript dropped if the meaning is clear,  be the
matrix corresponding to the linear equations given in
\eqref{eq:lin_prog_primal}; that is the $\sigma, r$ entry of $A_{\cS}$ is $1$ if
the atom $\sigma$ is in the sensor region $S_r$ and $0$ otherwise. The sensor measurement
vector is $\bn=(n_1,n_2,\ldots, n_R)$ and the atom value vector is
$\mathbf{m} =(m_a)$. Then the linear programming formulation gives that the minimal
value for this sensor configuration $\bS$ and sensor measurement vector $\bn$ is
\begin{equation}
  \label{eq:lin_prog_mat}
\minval(\bS,\bn)=\min\{\bone \cdot \mathbf{m}:  A \mathbf{m} = \bn, \quad \mathbf{m} \geq 0\},
\end{equation}
where $\cdot$ indicates the dot product of the two vectors, so that $\bone
. \mathbf{m}$ is the sum of the entries in $\mathbf{m}$.

The following theorem then addresses the range of values.
\begin{theorem}
  \label{thm:range} Given a sensor configuration $\bS=(S_1,S_2,\ldots,
  S_R)$ and sensor measurements  $\bn=(n_1,n_2,\ldots, n_R)$, every (real) value
  between the maximum and minimum overall  value  is achievable.
\end{theorem}

\begin{proof}
  The function $$f(\mathbf{m}) = \sum_{a\text{ atom}} m_a$$ is continous as well as
  defined on a bounded, convex polytope which is compact, and
  connected.  This function $f(\mathbf{m})$ then satisfies the
  intermediate value theorem.
\end{proof}

The integer case is more difficult (and interesting). We have the
following theorem.
\begin{theorem}
  \label{sec:3}
 If
  $\bS$ is generic, then for every set of sensor managements that are
  integers, every
  integer value between the minimum overall (integer) count and maximum overall
  count  is  achievable by  integer assignments to  atoms.
\end{theorem}

\begin{proof}
It is enough to observe  that if $\bn$ is the sensor measurement vector and $k$ is
  a possible achievable integer value for the overall count, then so
  is $k+1$ unless $k$ is the maximum possible count. Suppose first
  that there is no target in an intersection $S_i\cap S_j$ with $i\neq
  j$. Then $k$ is the maximum possible count. Otherwise, suppose that
  a target $t$ is in $S_i\cap S_{j_1}\cap\cdots \cap S_{j_r}$, and in
  no other sensor regions.

  This target is moved into $\bigl(S_{j_1}\cap S_{j_2}\cap \cdots\cap
  S_{j_r}\bigr)\backslash S_i$, which is possible because of
  genericity. Then insert another target in $S_i\backslash \bigcup_{i=1}^r
  S_{j_i}$. This keeps the sensor measurement  for every sensor region
  unchanged but increases the overall value  by $1$. Moreover, it is
  clear that integers are assigned to atoms by this proof.
\end{proof}

We note that the genericity condition is required here, though a
somewhat weaker one would suffice. The first configuration in
Figure~\ref{fig:degeneracy} fails to be generic but still has the
property that the range of integer overall  values   forms an interval for
any sensor measurement  vector. On the other hand, consider the sensor
configuration in Figure~\ref{fig:noninterval}.
\begin{figure}[h]
  \centering
\includegraphics[width=0.3\textwidth]{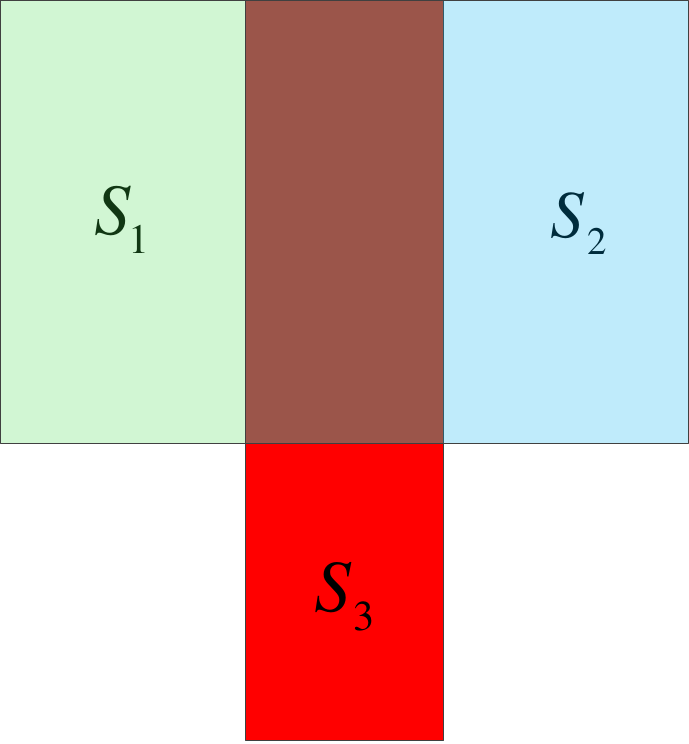}
    \caption{Example where the integer values do not form an interval.}
  \label{fig:noninterval}
\end{figure}
If each sensor region has a count of $1$, then $3$ is clearly the
maximum count and $1$ is the minimum, but there is no achievable
overall count of $2$ with integer sensor measurements.

It is clear that the maximum overall count associated with integer
measurements is the same as the maximum real value associated with
those measurements since this is just the sum of the individual sensor
counts. On the other hand, this is clearly not the case  for the
minimum overall value. To see this, consider the case of three sensor
regions $S_1, S_2, S_3$ with only pairwise non-empty intersections:
$S_i\cap S_j\neq \emptyset$ for all $i,j=1,2,3$. In this case, as we
shall see later, the minimum overall value is $\frac{n_1+n_2+n_3}{2}$
where $n_i$ is the  count in region $S_i$. As this is not always an
integer it cannot be achieved by integer atom values. One might ask
whether the smallest integer greater than or equal to this minimum
value is achievable by integer atom counts. We have been unable to
definitively answer this question, though our experiments suggest that
it is always true for generic sensor configurations.

In principle  \eqref{eq:lin_prog_mat} provides a mechanism for
calculation of the minimum, and therefore by Theorem~\ref{thm:range} the
entire range of possible values, but  for large numbers of sensors it
becomes impractical. Moreover it suffers the problem that, whenever the sensor
value changes, an entirely  new calculation is needed.  What is required is a formulation
of the problem that provides a simple machine to go from sensor measuremens to
minimal overall  value. This machine should itself be only dependent on the geometry of
the sensor regions, with only the inputs of sensor measurements  changing. In
other words, we would prefer a computable formula into  which
insertion of the  sensor measurements  presents the value of (\ref{eq:lin_prog_mat}).

Progress towards such  a solution, is obtained via  the dual linear programming
problem. Since the formulation in \eqref{eq:lin_prog_mat} is in standard form,
the dual problem is (see  \cite{dantzig98:_linear_progr_exten}, p.128) expressed in
terms of dual variables $\by=(y_r)$ indexed by sensor regions and states
\begin{equation}
  \label{eq:dual_prob}
  \max\{\bn\cdot\by:A^T\by\leq \bone\}.
\end{equation}
The key difference between this and the usual dual-primal formulation
(see \cite{dantzig98:_linear_progr_exten}) is that, because
the primal is stated in terms of equalities, the dual variables have no
restrictions other than the linear inequality $A^T\by\leq \bone$; in
particular, $\by$ is not required to be non-negative.  As a result the feasible
region (the convex set described by the constraints) is not, in general, compact. Existence
of solutions therefore becomes an issue.

As an illustration of the problems that can arise consider the case given in
Fig.~\ref{fig:degeneracy}(a). The matrix $A$ is obtained as
\begin{equation}
  \label{eq:a-degenerate}
  A=
  \begin{pmatrix}
    1&0&0&1&0&1\\
    0&1&0&0&1&1\\
    0&0&1&1&1&1\\
  \end{pmatrix},
\end{equation}
and the dual constraint becomes
\begin{equation}
  \label{eq:dual-a}
  A^T\by\leq 1.
\end{equation}
This has a solution $(1,1,-1)$ which gives the maximum for  $n_1=1,
n_2=1, n_3=1$ (and indeed for any values for which $n_3\leq \min(n_1,n_2)$). It
is easy to see that this is an extreme point of the convex (non-compact)
polytope specified by \eqref{eq:dual-a}.

Generally, having to deal with a dual polytope that is not non-negative causes
problems and leads to complications in calculating the minimum value.
However, many of the potential issues disappear because of the specific nature
of the problem. The feasible region is always non-empty:
$(\frac{1}{R},\frac{1}{R},\frac{1}{R},\ldots,\frac{1}{R})$ is always in it for
example, and the region is bounded above: $y_r\leq 1$ for all $r$.  It follows
by \cite{dantzig98:_linear_progr_exten} that there is a solution of the dual problem
and it equals the solution of the primal. It follows too that the minimum
value is exactly the solution of the dual problem, and that once the extreme
points of the dual polytope are found it is a simple matter to take their
inner products with the sensor measurement  vector $\bn$ and find its maximum to
calculate the minimum count. This means that the polytope, or rather its set
of extreme points, provides the computational machine needed. We call this
polytope the \emph{fusion polytope} for the given sensor configuration and
denote it by $\cC(\bS)$. It is, obviously, only dependent on the sensor
configuration. We state this key result as a theorem.
\begin{theorem}
  \label{thm:dual_result}
Given a sensor configuration $\bS$, with fusion polytope $\cC(\bS)$, for any
sensor measurement  vector $\bn$,
\begin{equation}
  \label{eq:adlnfnajnd}
  \minval(\bS,\bn)=\max\{ \bn\cdot\be: \be \text{ is an extreme point of $\cC(\bS)$}\}.
\end{equation}
\end{theorem}

Naturally, in the case of a counting, rather than a continuum problem,
the solution needs to be an integer, but the smallest integer greater
than the actual minimal real solution provides the minimum in that
case.  We reiterate that, by passage to the dual linear programming
problem, the methods to calculate minimum sensor measurements is reduced to
testing certain ``universal'' vertices against the sensor readings. Once
the extreme points of the fusion polytope are identified, a formula
for the minimum that involves sensor readings as variables is
immediate. Indeed there are some extreme points that are unnecessary
for this calculation: for instance if $\be$ and $\be'$ are extreme
points and $e_{i}'\geq e_{i}$ for all $i$ then there is no need to
include $\be$ in the maximization in (\ref{eq:adlnfnajnd}). We say
that $\bE$ is a \emph{dominant} extreme point if there is no such
$\be'$. Evidently it is only necessary to list all dominant extreme
points for the purposes of (\ref{eq:adlnfnajnd}).

Despite the existence of solutions to the dual problem, the
calculation of (dominant) extreme points of the fusion polytope
appears hard in general, and we currently have no solution that
applies across most situations.  Even in some simple cases resulting
from graphical models, studied by~\cite{Bruce2003Book}, where the
polytope is given by $x_i\geq 0$ and $\Sigma_{i\in Q} x_i \leq 1$ over
cliques $Q$, maximization of a linear objective function over
$\qstab(G)$ (see Section~\ref{sec:examples} for the definition) is
known to be NP-complete.

It turns out, somewhat surprisingly,  that the issue of  the fusion
polytope not residing in the positive orthant is exactly that of non-genericity.
\begin{theorem} \label{thm:generic}
Let the sensor configuration $\bS$ be generic, and let $\cC(\bS)$ denote the fusion polytope
specified by \eqref{eq:dual_prob}. Then
\begin{equation}\label{eq:ksjdb}
  \max_{\mathbf{c} \in \mathcal{C}(\mathbf{S})} \mathbf{n}\cdot\mathbf{c} = \max_{\mathbf{c} \in \mathcal{C}(\mathbf{S}) \cap P} \mathbf{n}\cdot\mathbf{c}
\end{equation}
where $P$ is the positive orthant: $P=\{\by: y_r\geq 0\ \text{for all
  $r$}\}$. The converse of this statement is true. If \eqref{eq:ksjdb}
is true then the sensor configuration is generic.
\end{theorem}

\begin{proof}
Fix $\bn$ and suppose that $\wy$ is an extreme point of $\mathcal{C}(\mathbf{S})$ for which
\begin{equation}
  \label{eq:kabdjka} \bn\cdot\wy= \max_{\bc\in \mathcal{C}(\mathbf{S})} \bn\cdot\bc,
\end{equation} and assume that $\wy\not \in P$.  Suppose further that there is
no extreme point of $\mathcal{C}(\mathbf{S})$ that is non-negative, at which the maximum is
achieved. We may assume, without loss of generality that the
first $s$ coordinates $\wy_1,\wy_2,\ldots,
\wy_s$ of $\wy$ are all negative and that all
subsequent coordinates of $\wy$ are
non-negative. Let $\wy'$ be the vector $\wy$ with the first $s$ coordinates
replaced by $0$.

Let $A$ be the defining matrix for the linear programming problem as
in \eqref{eq:lin_prog_mat}.
The generic property entails that if $\ba$ is a row of $A^T$ then a row vector  $\ba'$
for which the coordinatewise product satisfies $\ba\cdot\ba'=\ba'$ is also a row of
$A^T$. It follows that, for any row $\ba$  of $A^T$, there is a row of
$A^T$ which has the first $s$ rows of $\ba$ replaced by $0$. As a result,
$\ba\cdot\wy'\leq 1$ for all rows $\ba$ of $A$. Thus $\wy'\in \mathcal{C}(\mathbf{S}) \cap P$. Moreover,
clearly
\begin{equation}
  \label{eq:aldjna} \bn\cdot\wy'\geq \bn\cdot\wy.
\end{equation} The conclusion follows.

For the converse, we observe that the constraints in the dual problem
are labelled and determined by atoms. In consequence
$\cC(\bS_{\text{gen}})$ is a subset of $\cC(S)$. However, in the case
of $\cC(S)\cap P$, only the constraints that correspond to the highest
dimensional simplices (or equivalently the largest number of sensor
regions in the intersection forming the atom) are relevant, since any
constraint imposed in respect of such a simplex is stronger than one
corresponding to a simplex contained in it. The highest dimensional
constraints are common to both $\bS$ and to $\bS_{\text{gen}}$, so
that
\begin{equation}
  \label{eq:adnalkda}
  \cC(\bS)\cap P=  \cC(\bS_{\text{gen}})\cap P.
\end{equation}
Now, for a sensor measurement  vector $\bn$,
\begin{equation}
  \begin{aligned}
  \label{eq:abadakal}
  \max\{\by\cdot\bn:\by\in \cC(\bS)\}&\geq   \max\{\by\cdot\bn:\by\in \cC(\bS)\cap P\}\\
&=\max\{\by\cdot\bn:\by\in \cC(\bS_{\text{gen}})\cap P\}\\ &=
\max\{\by\cdot\bn:\by\in \cC(\bS_{\text{gen}})\}.
  \end{aligned}
\end{equation}
Thus if
\begin{equation}
  \label{eq:abcaoidn}
  \max\{\by\cdot\bn:\by\in \cC(\bS)\}=\max\{\by\cdot\bn:\by\in \cC(\bS)\cap P\},
\end{equation}
then
\begin{equation}
  \label{eq:adnalkdaalkd}
  \max\{\by\cdot\bn:\by\in \cC(\bS)\}=
\max\{\by\cdot\bn:\by\in \cC(\bS_{\text{gen}})\}.
\end{equation}
If $\bS$ is not generic, then there is some intersection $S_{i_1}\cap
S_{i_2}\cap\cdots  \cap S_{i_R}$ contained in a union of some different sensor regions
$S_{j_1},S_{j_2},\ldots, S_{j_T}$.
Consider a sensor measurement  vector which assigns $1$ to each of the sensor
regions $S_{i_r}$ $(r=1,2, \ldots, R)$, and $0$ to every other sensor region.

The minimal overall value  for $\bS$ is at least $2$ because there can be no targets in
the intersection  $S_{i_1}\cap
S_{i_2}\cap\cdots  \cap S_{i_R}$. On the other hand for the corresponding
sensor measurement  in the generic case the minimal overall value  is
$1$, since there can be just one target in the intersection $S_{i_1}\cap
S_{i_2}\cap\cdots  \cap S_{i_R}$.
\end{proof}

This result considerably simplifies calculations in the generic case. We call
$\mathcal{C}(\mathbf{S}) \cap P$ the \emph{positive fusion polytope},
though the word ``positive'' will be dropped where there is no
likelihood of confusion. One simple consequence of this result is that
in the generic case, if the sensor measurement  increases; that is, if we
have two sensor measurement  vectors $\bn=(n_1,n_2,\ldots,n_R)$ and
$\bn'=(n'_1,n'_2,\ldots,n'_R)$ where $n_i\geq n_i'$, then the minimum
increases, or at least does not decrease. This may seem, at first
sight  obvious, but it is quickly apparent from Theorem~\ref{thm:generic}
that it is  not true in the degenerate case, and indeed this is a defining
characteristic of genericity.  As  simple example consider  the case of
Figure~\ref{fig:degeneracy}(b). If the sensor measurement  vector is
$(1,0,1)$, then the minimum overall value  is $2$, whereas if the sensor measurement  vector
is $(1,1,1)$ the minimum overall value  is $1$.

In the generic case, the linear programming
problem is expressible entirely in terms of the simplicial complex.
The translation is as follows. For each vertex
of the simplicial complex (sensor region) we have the following constraint:
\begin{equation}
  \label{eq:2}
  \sum_{\sigma\text{ contains vertex $v$}}m_{\sigma}=n_{v},
\end{equation}
corresponding to the constraint \eqref{eq:lin_prog_primal} in the ``region''
formulation.  In addition, of course $m_{\sigma}\geq 0$. We need to minimize
\begin{equation}
  \label{eq:3}
  \sum_{\sigma} m_{\sigma},
\end{equation}
subject to these constraints.
The dualization results in the following problem expressed in terms of the
geometry of the simplicial complex.
\begin{equation}
  \label{eq:lakdnfd}
  \max\{\sum_r y_r:\sum_{r\in \sigma}y_r\leq 1,\ (\sigma\in \Sigma),\
  y_r\geq 0\  (r=1,2,\ldots,R)\}.
\end{equation}
The fusion polytope is then, with some abuse of notation,
\begin{equation}
  \label{eq:adadoob}
  \cC(\Sigma)=\{\by=(y_r):\sum_{r\in \sigma}y_r\leq 1,\ (\sigma\in
  \Sigma),\ y_r\geq 0 \ (r=1,2,\ldots,R)\}.
\end{equation}

Observe that if $\sigma\subset \sigma'$ are simplices then $\sum_{r\in
  \sigma}y_r\leq \sum_{r\in \sigma'}y_r$ so that the only inequalities that
need be considered in defining $\cC(\Sigma)$ are those that are maximal; that
is, are not contained in any larger simplex. Thus in  the case of
three regions with all possible intersections, so that the simplicial
complex is the triangle and all of its subsimplices,  the
only inequality (other than non-negativity) that is needed is $y_1+y_2+y_3\leq
1$. It follows immediately that the extreme points of $\Sigma$ in this case
are $(1,0,0),\ (0,1,0),\ (0,0,1)$ and the minimal count is $\max(n_1,n_2,
n_3)$.

As another example consider the case where the simplicial complex $\Sigma$
consists of the faces of a tetrahedron. This corresponds to four regions such
that each triple intersection is non-empty but the quadruple intersection is
empty. The constraint matrix  for the dual problem (corresponding only to highest dimensional
simplices) is
\begin{equation}
  \label{eq:albcfabd}
  A^T=
  \begin{pmatrix}
    1&1&1&0\\
    1&1&0&1\\
    1&0&1&1\\
    0&1&1&1\\
  \end{pmatrix}.
\end{equation}
It is fairly easy to see that the fusion polytope in this case has extreme
points $(1,0,0,0)$, $(0,1,0,0)$, $(0,0,1,0)$, $(0,0,0,1)$ and $(\frac{1}{3},
\frac{1}{3}, \frac{1}{3}, \frac{1}{3})$, and so the minimum count is
\begin{equation}
  \label{eq:agadbala}
  \max(n_1,n_2,n_3,n_4,\frac{1}{3}(n_1+n_2+n_3+n_4)).
\end{equation}

\section{Graphs}\label{sec:examples}

There are  a number of cases where the  minimum
estimation problem can be posed in terms of constructs on graphs. This happens
in various ways. Even when such a reformulation is possible, it does
not always provide a feasible approach to computation of the
minimum; rather these formulations and the effort in the
graph theory community to solve the corresponding problems suggest
that the minimum estimation problem is difficult.

%%% Here talk about STAB etc FRAC(G)
\subsection{The Fractional Stable Set}
\label{sec:graph}
To describe the ideas, we consider a graph $\cG=(\cV,\cE)$. A set of
vertices $U\subset \cV$ is called a \emph{stable set} if no two
distinct elements of $U$ have an edge joining them; that is, if
$u_1\neq u_2$, $u_1,u_2\in U$, then $(u_1,u_2)\not \in \cE$. The convex
hull
\begin{equation}
  \label{eq:STAB}
  \stab(G)=\co \{\bone_U: U \text{ is a stable set of vertices}\}.
\end{equation}
is called the \emph{stable set polytope}. The \emph{fractional stable
  set polytope} of a graph is
\begin{equation}
  \label{eq:5}
 \fract(G) = \{ (x_v)_{v\in \cV}: x_v\geq 0,\ x_{v_1}+x_{v_2}\leq 1, \
  (v,v_1,v_2 \in \cV), \ (v_1,v_2)\in \cE\}.
\end{equation}
It is straightforward, and well known that $\stab(G)\subset
\fract(G)$, and that the latter is exactly the fusion polytope for the
situation where the $R$ sensor regions correspond to the vertices of the
graph, pairwise intersections correspond to edges, and there are no
triple or higher intersections; that is, where the nerve is just a
graph.  In this context some results exist; a
good reference is the notes of Wagler
\cite{wagler03:_graph_theor_probl_relat_polyt}, but even this
relatively simple case appears to produce no simple algorithm for
computation of the extreme points. Here is an important theorem.

\begin{theorem}[Gr\"otschel, Lov\'asz, Schrijver, \cite{groetschel88:_geomet_algor_combin_optim}] \label{thm:1} $\ $
  \begin{enumerate}
  \item $\stab(G)=\fract(G)$ if and only if $G$ is bipartite with no isolated vertices;
  \item The extreme points  of $\stab(G)$, regarded as functions on
    the vertices of $G$,  all have values $0$ and $1$;
  \item The extreme points of $\fract(G)$ all have values $0$, $1/2$, $1$.
  \end{enumerate}
\end{theorem}
%Also

We illustrate these ideas with the following simple example.
\begin{example}
  \label{ex:necklace}
We consider the case when the regions $S_{1}, S_{2}, \ldots, S_{R}$ satisfy
$S_{i}\cap S_{i+1}\neq \emptyset$, $S_{1}\cap S_{R}\neq \emptyset$
and these are the only non-empty
 intersections. The nerve is a graph with $R$ vertices connected in a cycle.
 A collection of extreme points involve just $0$s and $1$s which
 satisfy the following rules:
 \begin{enumerate}
 \item every $1$ is isolated, that is, $11$ does not appear;
 \item no sequences with three adjacent $0$s appear.
 \end{enumerate}
Thus, between every pair of $1$s either a single $0$ or a pair of $0$s
appears. If $R$ is even, all extreme points are of this form, but if
$R$ is odd there is an additional extreme point of the form
$\bd=(\half,\half,\ldots, \half)$.
When $R$ is even $\bd$ is a convex combination of two $0$-$1$  extreme
points and so is not extreme.
\end{example}

While the preceding discussion shows how graph theoretic ideas apply
for the case where there are no non-trivial triple intersections,
there are some other situations where graph theory is applicable.
A construct  on graphs related to the stable and the fractional stable
set polytopes is obtained as follows. We recall that a \emph{clique}
in the graph $G$ is a set $Q$ of vertices that as the induced  subgraph of $G$ is
complete; that is, each pair of vertices in $Q$ is connected by an edge.
Now we define
\begin{equation}
  \label{eq:4}
  \qstab(G)=\{ (x_v)_{v\in \cV}: x_v\geq 0,\ \sum_{v\in Q}x_{v}\leq 1, \
  \text{$Q$ is a clique in $G$}\}
\end{equation}

The $\qstab$ construction permits representation of the fusion
polytope for  other sensor configurations than just those with only
pairwise intersections. For instance,  consider
Figure~\ref{fig_qstab1}. In this case, the $2$-simplices are present
and so the nerve  is not a graph. However, the cliques of
the $1$-skeleton $G$  are just the two $2$-simplices and so
$\qstab(G)$ is the fusion polytope. On the other hand if the sensor
regions are as in Figure~\ref{fig_qstab2} then the simplicial complex
is $1$-dimensional and comprises the edges of a triangle. In this case
the single clique consists of all vertices of the graph and so
$\qstab$ is not the same as the fusion polytope.

A \emph{flag complex} is an abstract simplicial complex in which every
minimal nonface has exactly two elements~\cite{Tits1974Book}.  In
essence, any flag complex is the clique complex (that is, the set of
all cliques) of its $1$-skeleton.  For example, if three edges of a
$2$-simplex are in the simplicial complex in question, then so is the
$2$-simplex itself.  More generally, if all of the faces of a simplex
are in the simplicial complex, then so is the entire simplex.  The
boundary of a $2$-simplex  does not form a flag complex, but the $2$-simplex
itself does.
%For instance, since set $\{1,4\}$ is the only minimal nonface in Figure~\ref{fig:qstab1}, the simplicial complex is a flag complex.
%However, $\{1,4\}, \{1,2,3\}, \{2,3,4\}$ are minimal nonfaces in Figure~\ref{fig:qstab2}, thus it is not a flag complex.
We have the following theorem.
\begin{theorem}\label{thm:2}
Given sensor regions $S_{1},S_{2},\ldots, S_{R}$ which form a generic
sensor configuration, the fusion polytope is identified with
$\qstab(G)$ via the inclusion of a graph $G$ as the $1$-skeleton of
the simplicial complex of the sensor configuration
if and only if it  is a flag complex.
\end{theorem}
\begin{proof}
  We assume throughout that the sensor configuration is generic.
  Next, assume that the simplicial complex of the sensor configuration
  is a flag complex $K$. Then $K$ is the clique complex of its
  $1$-skeleton $G$ by definition.   $\qstab(G)$ is determined by
  all of the cliques and the fusion polytope is determined by the simplicial
  complex. Since this is  a clique complex by definition,  the fusion
  polytope is identified with $\qstab(G)$.

 For the
  reverse direction, assume that the fusion polytope is identified
  with $\qstab(G)$, and that $G$ is the $1$-skeleton of the simplicial
  complex of the sensor configuration.  Since the fusion polytope is
  determined by this simplicial complex, and
  $\qstab(G)$ is determined by all cliques, there is a
  correspondence between simplicies and cliques, which means any
  simplex of the simplicial complex is a clique of $G$.
  The simplicial complex is then a flag complex by definition.
\end{proof}

Much work has been done on $\qstab(G)$ by many authors; see, for
instance, \cite{KosterWagler2006} for several references.  Since we
are interested in the extreme points of $\qstab(G)$ in situations when
$\qstab(G)$ is the fusion polytope of a sensor configuration, the
following result is important. It will require the definition of a
\emph{perfect} graph. First we need to know that a subgraph of $G$ is
\emph{induced} if is obtained by taking a subset of the vertices of
$G$ and all edges from $G$ with endpoints in the subset. A graph $G$
is \emph{perfect} if, for every subgraph $H$ of $G$, the chromatic
number of $H$ is equal to the size of the largest clique of $H$. All
bipartite graphs are perfect, whereas an odd cycle is not.

\begin{theorem}[\cite{Chvatal1975Journal}]
  \label{thm:qstab_perfect}
  Let $G$ be a graph. Then $\stab(G)=\qstab(G)$ if and only if $G$ is
  perfect.
\end{theorem}

Of course, if $\stab(G)=\qstab(G)$ then the dominant extreme points must all
take the values $0$ and $1$ and are fairly easy to write down; indeed
they are the characteristic functions of maximal stable sets.

\def\firstcircle{(0,0) circle (1.5cm)}
\def\secondcircle{(45:2cm) circle (1.5cm)}
\def\thirdcircle{(0:2cm) circle (1.5cm)}
\def\fourthcircle{(22.5:4cm) circle (1.5cm)}

% Now we can draw the sets:
\begin{figure}[h!]
  \centering
\begin{tikzpicture}
    \draw[fill=green,opacity=0.4] \firstcircle node[below,opacity=1] {$S_{1}$};
    \draw[fill=blue,opacity=0.4] \secondcircle node [above,opacity=1] {$S_{2}$};
    \draw[fill=yellow,opacity=0.4] \thirdcircle node [below,opacity=1] {$S_{3}$};
    \draw[fill=red,opacity=0.4] \fourthcircle node [above,opacity=1] {$S_{4}$};
\end{tikzpicture}
\begin{tikzpicture}[-,shorten >=1pt,auto,node distance=3cm,
  thick,main node/.style={circle,fill=blue!20,draw,font=\sffamily\Large\bfseries}]

  \node[main node] (s1) {$S_{1}$} ;
  \node[main node] (s2) [above right of=s1] {$S_{2}$};
  \node[main node] (s3) [right of =s1] {$S_{3}$};
  \node[main node] (s4) [right of=s2] {$S_{4}$};

  \path[every node/.style={font=\sffamily\small}]
    (s1) edge node [left] {} (s2)
    (s1) edge node [right] {} (s3)
    (s2) edge node [right] {} (s3)
    (s2) edge node [left] {} (s4)
    (s3) edge node [left] {} (s4);
  \end{tikzpicture}
  \caption{Example where $QSTAB(G)$ is the Fusion Polytope}
  \label{fig_qstab1}
\end{figure}
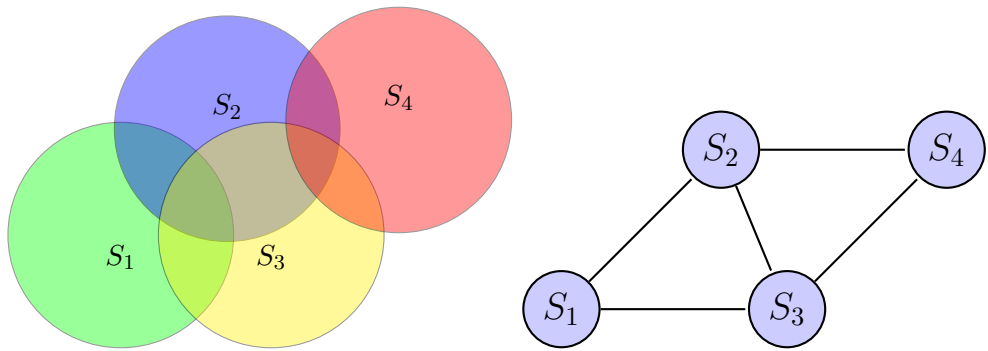

\def\firstellipse{(0,0) ellipse (1.5cm and 0.5cm)}
\def\secondellipse{(70:1.2cm) ellipse (1.5cm and 0.5cm)}
\def\thirdellipse{(-40:1.2cm) ellipse (1.5cm and 0.5cm)}

% Now we can draw the sets:
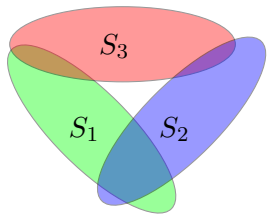
\begin{figure}[h!]
  \centering
\begin{tikzpicture}
    \draw[rotate=-45] \firstellipse[fill=green,opacity=0.4];
    \draw \secondellipse[fill=red,opacity=0.4];
    \draw[rotate=45] \thirdellipse[fill=blue,opacity=0.4];
      \draw (0.3,1.1) node {$S_{3}$};
  \draw (1.1,0) node {$S_{2}$};
    \draw (-0.1,0) node {$S_{1}$};
  \end{tikzpicture}
  \caption{Example where $QSTAB(G)$ is different from the  Fusion Polytope}
  \label{fig_qstab2}
\end{figure}
\section{The Two Dimensional Case}\label{sec:1}
\subsection{Quadruple Intersections}
\label{sec:quad_int}
At this point we consider planar sensor regions of various kinds. As
an example, consider an ``inflated tiling'' of the plane; that is, a
standard regular tiling by rectangles, illustrated in
Figure~\ref{fig:23232}, where the tiles are slightly inflated to
create overlapping regions.
\begin{figure}[h!]
  \label{fig:bricks}
  \centering
\includegraphics[width=0.4\textwidth]{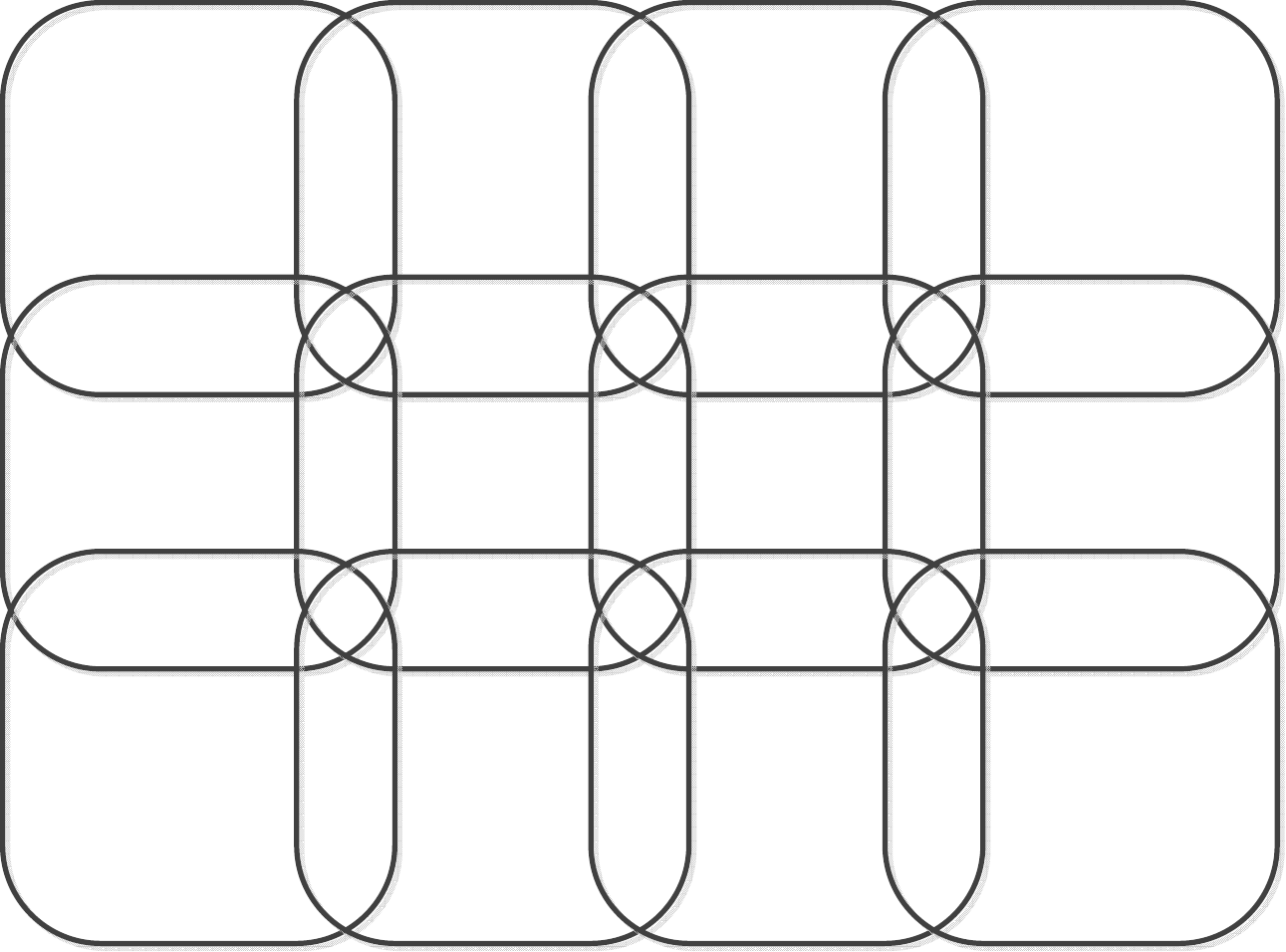}
  \caption{Tiling by  Rectangles}
  \label{fig:23232}
\end{figure}
To be precise,  the unit square $[0,1]\times[0,1]$ in
the plane is covered by $Q^2$ regions as in the figure. Note that there are quadruple
intersections but no higher, so that the nerve will be
$3$-dimensional. The regions are labelled $(S_{ij})_{i,j=1}^Q$ in
matrix style, and  the following quadruple intersections, and only
these are non-empty:
\begin{equation}
  \label{eq:1}
  S_{ij}\cap S_{(i+1)j}\cap S_{i(j+1)} \cap S_{(i+1)(j+1)}.
  \end{equation}
In addition double and triple intersections forced by these
quadruple ones are present. One might expect that this kind of sensor
configuration would be relatively easily to handle. Unfortunately it fails to be generic, specifically there are relations such as:
\begin{equation}
  \label{eq:andaldl}
  S_{11}\cap S_{22}\subset S_{12}\cup S_{21},
\end{equation}
so that the potential atom $S_{11}\cap S_{22}\cap S_{12}^{c}\cap
S_{21}^{c}\cap \bigcap_{i>2 \text{ or } j>2}S_{ij}^{c}$ is empty.
It follows, by Theorem~\ref{thm:generic}, that the fusion polytope
is not in the positive orthant. Indeed the point $(-1,1,1,-1)$ is an extreme
point of the fusion polytope for the case of just four regions
$S_{11},S_{12},S_{21},S_{22}$.

It turns out that this situation is typical. Indeed, any configuration
of four regions in the plane that satisfies some fairly reasonable
topological assumptions and that has a quadruple intersection fails to
be generic, as the following theorem shows. In order to describe the
result, more precision will be needed about the nature of regions and
their intersections than has been required so far. At this point we
shall introduce the topological constraints predicted when atoms were
first defined in \ref{def:atom}. We need to make some minor
adjustments to definitions in order to enable us to use topological
arguments.

For this purpose, a \emph{regular} sensor region $S$ is a bounded open
set in the plane of which the boundary $\partial S$ is a simple
piecewise smooth closed curve. A \emph{regular} sensor configuration
is one comprising regular sensor regions. We need also to redefine the
notion of an \emph{atom}. In this case, we replace each set of
the form $S_{i}^{c}$ in (\ref{eq:2}) by its interior $\overline{S_{i}}^{c}$,
but we still demand that an atom is non-empty, though we will use the
phrase that the ``atom is represented'' to mean exactly that. Note
that atoms are always open with this definition.
A regular sensor configuration is said to
be \emph{generic} if an intersection of sensor regions is not
contained in a union of different sensor regions.

A regular sensor configuration $\{S_{1}, S_{2}, S_{3},\ldots, S_{R}\}$
with $X=\cup_{r=1}^{R} S_{r}$ is
said to be \emph{normal} if the boundaries of any two sensor regions intersect
in at most two points, and there are no points of triple intersection
of boundaries of sensor regions. We also require that boundaries of
atoms are simple closed curves; these are, of course made up of
segments of boundaries of sensor regions.
Observe that, by the Jordan-Schoenflies Theorem
\cite{cairns51:_elemen_proof_jordan_schoen_theor}, these conditions
imply that both closures of sensor regions and closures of atoms are
homeomorphic to closed discs.
For the rest of this section we assume that all sensor regions are regular and just refer to
them as sensor regions.
\begin{theorem}
  \label{thm:quad_intersect}
  For a normal sensor configuration $\mathbb S=\{S_1,S_2, S_3, S_4\}$
  with closure a simply-connected closed set in the plane $\mathbb
  R^2$,
  at least two atoms are not represented; that is, are empty. The
  sensor configuration is not generic.
\end{theorem}
\begin{proof}
  A simplicial complex is constructed with the intersection points of
  boundaries of sensor regions as vertices, the segments of the
  boundaries between the intersection points as edges and the atoms of
  the sensor configuration as faces. This is a planar simplicial
  complex.  By our assumption on intersections of boundaries of  sensor regions,
there are at most $12$ vertices~($6+4+2$) and the degree of
  each vertex is $4$ in the graph which is the $2$-skeleton of the
  simplicial complex.  Furthermore, the number of
  edges is twice of the number of vertices, by the condition on the
  boundaries of atoms.  Since the Euler characteristic of the closure
  of the union of the sensor regions, is $1$,    $v - e + f = 1$, we have $v = f -
  1$.  Since $v \leq 12$, $f \leq 13$ is obtained; that is, there are at
  least two atoms not represented.
\end{proof}

If the condition of normality is relaxed to allow the boundary of one
sensor region to intersect the boundary of
another  in $4$ points, then
there will be two more vertices in the corresponding graph and the
number of possible faces will be equal to $15$.  All atoms can then be represented
and connected, as the example in
Figure~\ref{fig:four_regions_all_rep_connected} demonstrates.
Furthermore, if more pairwise intersections of boundaries of sensor
regions with $4$ points appear, all atoms still can still be  represented,
however, will be disconnected as shown in
Figure~\ref{fig:four_regions_all_rep}.

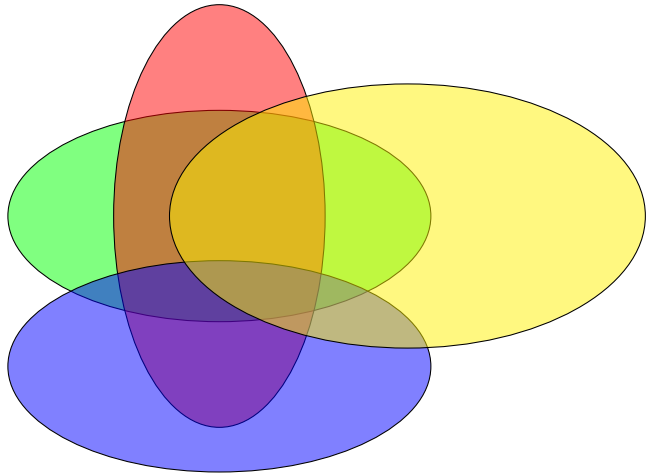
\begin{figure}[h!]
  \centering
  \begin{tikzpicture}
  \begin{scope}[fill opacity=0.5]
      \draw[rotate=0,fill=green] (0,0) ellipse (80pt and 40pt);
  \draw[rotate=90,fill=red] (0,0) ellipse (80pt and 40pt);
  \draw[rotate=0,fill=blue] (0,-2) ellipse (80pt and 40pt);
  \draw[rotate=00,fill=yellow] (2.5,0) ellipse (90pt and 50pt);
  \end{scope}
\end{tikzpicture}
  \caption{Four Regions in the Plane  with All  Atoms Represented and Connected}
  \label{fig:four_regions_all_rep_connected}
\end{figure}

\begin{figure}[h!]
  \centering
  \begin{tikzpicture}
  \begin{scope}[fill opacity=0.5]
      \draw[rotate=0,fill=green] (0,0) ellipse (80pt and 40pt);
  \draw[rotate=90,fill=red] (0,0) ellipse (80pt and 40pt);
  \draw[rotate=0,fill=blue] (0,-1) ellipse (80pt and 40pt);
  \draw[rotate=00,fill=yellow] (2,0) ellipse (110pt and 50pt);
  \end{scope}
\end{tikzpicture}
  \caption{Four Regions in the Plane  with All  Atoms Represented}
  \label{fig:four_regions_all_rep}
\end{figure}
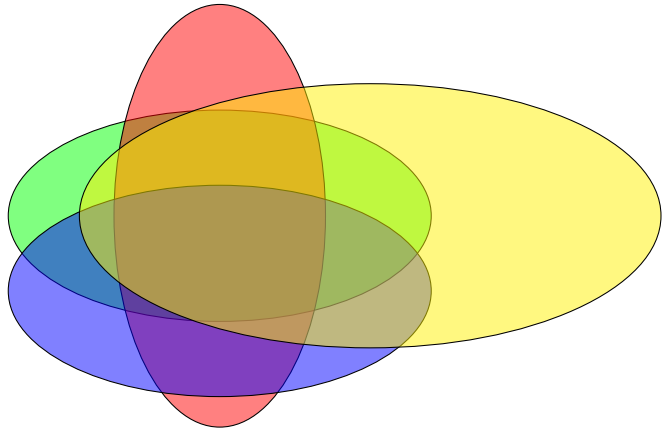

Despite this failure of genericity for the inflated rectangular
tilings it is still possible to use the dual linear programming
technique in this case. As already noted, the fusion polytope
is no longer in the positive orthant and the simplifications
associated with that property are not available. It is still possible,
however, to place a lower bound on the polytope as the following
result indicates, and the region need not be rectangular for this to
happen. The following notation will be useful:
\begin{equation}
  \label{eq:weps}
  D_\epsilon=\{\bx\in \bR^2:\|\bx-I^2\|<\epsilon\}
\end{equation}
and $I^2=[0,1]^2$ is the unit square. We can use any one of a number
of (convex) norms, the key features being that
\begin{equation}
  \label{eq:intersects_weps}
  D_\epsilon\cap \bigl((0,1)+D_\epsilon\bigr)\cap \bigl((1,0)+D_\epsilon\bigr)\cap
  \bigl((1,1)+D_e\bigr)\neq \emptyset.
\end{equation}
Also $\epsilon$ needs to be small enough to prohibit intersections of
the form \[D_\epsilon\cap \bigl((2,0)+D_\epsilon\bigr)\text{ or }D_\epsilon\cap \bigl((0,2)+D_\epsilon\bigr).\]
\begin{theorem}
  \label{thm:squares_minus_one}
  Let $T$ be a region in the plane that is a union of inflated squares
  based on the integer lattice. Thus
  \begin{equation}
    \label{eq:squares_union}
    T=\bigcup_{r=1}^R (m_r,n_r)+D_{\epsilon},
  \end{equation}
and let the sensor regions be $(m_r,n_r)+D_\epsilon$. Then no point in
the corresponding fusion polytope has a coordinate less than $-1$.
\end{theorem}

\begin{proof}
The constraints of the fusion polytope correspond to non-empty intersections
of the sensor regions, and these are of four kinds (only the
translating integer pairs are listed):
\begin{itemize}
  \item \emph{Squares}: $\sigma=\{(n, m), (n,m+1), (n+1,m), (n+1,m+1)\}$.
  \item \emph{Corners}: $\ell=\{(n, m), (n,m+1), (n+1,m)\}$ and all
    variants of this by rotations through $90^\circ$ and reflections.
  \item \emph{Segments}:$\gamma=\{(n, m), (n,m+1)\}$ and all
    rotations through $90^\circ$.
  \item \emph{Singletons:} $\alpha=\{(n,m)\}$.
\end{itemize}
Thus, for a square we obtain the constraint,
\begin{equation}
  \label{eq:constraint_square}
  x_{(n, m)}+x_{(n,m+1)}+x_{(n+1,m)}+x_{(n+1,m+1)}\leq 1,
\end{equation}
with similar constraints for the corners and segments. Singletons give rise
to the constraint $x_{(n,m)}\leq 1$.

Write $A$ for the matrix whose rows are obtained from the
constraints. The columns of $A$ are indexed by $r=1,\ldots, R$. The
fusion polytope is then $C_F = \{\bx\in \bR^2:A\bx\leq 1\}$.

 The extreme points of $C_F$ are all obtained as the unique
   solutions  of equations of the form $A_1 \bx=\bone$ where $A_1$ is
   obtained by choosing rows from $A$ that form a basis for $\bR^N$.
   In addition, of course,  $\bx\in C_F$, so $A_1^c\bx\leq \bone$ for
   the remaining rows $A_1^c$ of $A$. We write $W(A_1)$ for the
   collection of squares, corners, segments and singletons
   corresponding to the rows of $A_1$.

   Now suppose that the $(n_r,m_r)$th coordinate is less than
       $-1$ and let $R$ be a region in $W(A_1)$  containing
       $(n_r,m_r)$. Clearly, $R$ cannot be a singleton.  Neither can it
       be a segment or a corner since the sum over $R$ has to be $1$ and
       each $x_{(n,m)}\leq 1$. Suppose then that $R$ is a square. This
       would imply that the sum over the members of the square other
       than $(n_r,m_r)$ is at least $2$. But these form a corner and so
       the sum has to be less than $1$.   This provides a contradiction and
       proves the result.
\end{proof}

Even though $C_F$ is not compact we can still use its extreme points to
count. The following result follows quickly from standard results on
convex polyhedra.

 \begin{lemma}
     If $\bc$ is a vector in $\bR^N$ in the positive quadrant
         then
         \begin{equation}
           \label{eq:akjdb}
           \sup \{ \bc\cdot\bx:\bx\in C_F\}=\max\{\bc\cdot\be:\be\in \ext C_F\}
         \end{equation}
         where $\ext C_F$ is the set of extreme points.
     \end{lemma}

\subsection{Hexagons}
\label{sec:hex}

If instead of using inflated rectangular tilings, we employ inflated
hexagonal tilings, the genericity problems disappear. Let $H$ be the regular hexagon
centred at the origin in $\bR^2$ and with vertices at
$(0,\pm 1)$,$(\pm \frac{1}{2},\pm\frac{\sqrt{3}}{2})$ and consider
 translates of this via the hexagonal lattice.
These tesselate the plane in a honeycomb arrangement.  Slight
inflations of them have at most $3$-fold intersections and so do not
suffer the problems of the rectangular tilings. In this case for a
cover of some region in the plane by translates of these inflated
hexagons the fusion polytope lies in the positive orthant. Even in
this case, though, for large numbers of such inflated hexagons the
fusion polytope can become very complicated and its number of
extreme points large. We illustrate this with a few examples computed
using the polymake package (\texttt{http://www.polymake.org/doku.php}). This is a topic we intend to return to in
a later paper.

Consider the  tesselation of a region of the plane in
Figure~\ref{fig:31hex}. In fact, of course, we are interested in
slight inflations of these hexagons so that the regions overlap. There
are obvious triple intersections but no quadruple intersection.

\input{white_31_regions}

This, incidentally, is almost as large a collection of hexagons as we are
reasonably able to handle on a laptop computer in the space of a few
days using  \texttt{polymake}.
According to polymake, the fusion polytope for the
collection of 30 hexagons depicted here has nearly 5000 dominant extreme points. Many,
in fact most, of these are extreme points involving only $0$ (white),
and $1$ (blue)
coefficients such as the one in Figure~\ref{fig:white_blue_31hex}. The
rule for generating such (dominant) extreme points is fairly simple. No two
adjacent hexagons can have a coefficient $1$, and the number of $1$s
is maximal subject to this constraint.

More interesting are extreme points with coefficients of $1/2$ (yellow) and
$0$ (white), such as in Figure~\ref{fig:white_yellow_31hex}.

\input{white_yellow_31_regions}

Here again it would not be too hard to develop a description of the
patterns. It is clear that no three hexagons that meet at a point can
all have coefficient $1/2$. Moreover, and these appear to be related
to the result \ref{thm:1} \cite{groetschel88:_geomet_algor_combin_optim}], every example of this kind involves a loop with an
odd number of hexagons.  There are, however, other possibilities for
extreme points.  One rather obvious one is a
mixture of the preceding two types as is illustrated in
Figure~\ref{fig:whiteblueyellow_hex}.

More exotic extreme points also arise. To illustrate, consider
Figures~\ref{fig:greyblackyellow} (where grey represents $\frac{1}{4}$
and black $\frac{3}{4}$, and yellow, as before, represents
$\frac{1}{2}$), and Figures~\ref{fig:whitegreenred} (where green represents
$\frac{1}{3}$ and red $\frac{2}{3}$). Coefficients other than
$0,1,\frac{1}{2}, \frac{1}{3},\frac{2}{3}, \frac{1}{4}, \frac{3}{4}$
have also been observed. We will discuss these patterns in much more detail in a later paper.
\input{greyblackyellow_30regions}

\section{Bounds and Special Types of Regions}
\label{sec:bounds}

We consider here only (irredundant, covering,) generic sensor
configurations. These are entirely represented by their nerves and so
are discussed in those terms. Accordingly, we fix a simplicial complex
$\Sigma$ which is the nerve of a sensor configuration $\bS=(S_1,S_2,\ldots,
S_R)$.  As before,  $\minval(\Sigma,\bn)$ denotes the minimum value
associated with this configuration and the sensor measurement  vector $\bn$.

The \emph{$k$-skeleton}  of $\Sigma$ is defined to be the subcomplex of $\Sigma$
consisting of those subsets of size $\leq k+1$ (or in other words simplices of
dimension $\leq k$ and denoted by $\Sigma^k$.
We write $\Sigma_{k}$ for  the simplicial complex obtained from
$\Sigma$ by inserting a new simplex whenever $\Sigma$ contains all
of the $k$-faces of that simplex. In other words, if every subset of $\sigma=(i_1,i_2,
\ldots, i_r)$ ($1\leq i_n \leq R$) of size $k+1$ belongs to $\Sigma$ then
$\sigma\in \Sigma$.

The following result is  a straightforward observation. Recall that
genericity permits us to consider the fusion polytope as a subset of
the positive orthant.
\begin{theorem}
\begin{equation}
  \label{eq:33}
  \cC(\Sigma_{1})\subset  \cC(\Sigma_{k}) \subset \cC(\Sigma)
\subset \cC(\Sigma^{k}) \subset \cC(\Sigma^{1}),
\end{equation}
and so for any measurement vector $\bn$,
\begin{equation}
  \label{eq:34}
 \max\bn= \minval(\Sigma_{1},\bn)\leq   \minval(\Sigma_{k},\bn)\leq   \minval(\Sigma,\bn)\leq
  \minval(\Sigma^{k},\bn)\leq   \minval(\Sigma^{1},\bn)= \bone.\bn.
\end{equation}
\end{theorem}
Of course $\Sigma^2$ is a graph and so the techniques discussed in
Section~\ref{sec:examples} can be applied to the calculation of
$\minval(\Sigma^2,\bn)$.

Some kinds of restrictions on shapes make calculating and bounding the minimal
count easier. When the regions are convex subsets of $\bR^{n}$,
Helly's Theorem is effectively as follows.
\begin{theorem}[\cite{helly23:_ueber_mengen_koerp_punkt}]
\label{thm:helly}
If $C_{1}, C_{2}, \ldots, C_{R}$ are convex subsets of $\bR^{n}$ with the
property that $C_{i_{1}}\cap C_{i_{2}}\cap\cdots\cap C_{i_{n+1}}\neq \emptyset$
for all choices of $i_{k}$, then $C_{1}\cap C_{2}\cap\cdots \cap C_{R}\neq
\emptyset$.
\end{theorem}
In terms of the simplicial complex associated with the sensor regions
this yields the following property.

\begin{theorem}
  \label{thm:simphelly}
  Let $\bS=(S_{1},S_{2},\ldots, S_{R})$ be a generic
  sensor configuration for which the sensor regions are convex subsets of
  $\bR^{N}$. Let $\Sigma$ be the corresponding simplicial complex. Then
   $\Sigma$ contains all of the $N$-dimensional faces of a simplex it must also contain
  the simplex. In other words $\Sigma_N=\Sigma$.
\end{theorem}
In particular, if  the sensor regions are in the plane $\bR^2$ then $\Sigma_2=\Sigma$.

Another situation in which it is possible to make further progress
concerns a generic collection of sensor regions with the property that every
sensor region contains a ``uniformly maximal dimensional
intersection''; in other words, there is some $M$ such that for every
$r$, $$S_r\supset S_{i_1}\cap S_{i_2}\cap \ldots \cap
S_{i_{M+1}}\neq \emptyset,$$ but there are no non-trivial non-empty
$M+2$ intersections.   This amounts to a ``manifold''
assumption on the nerve; namely, that every simplex is a subset of one
of maximal dimension $M$. In this case, in the dual formulation of the
linear programming problem, the only inequalities that need to be
considered are the ones involving the maximal dimensional simplices.
In this case, $\cC(\Sigma)=\cC(\Sigma_M)$, but more importantly, the
linear programming problem  is significantly simplified because of the
reduced number of inequalities. To illustrate this we present the
following example.

We describe several generic sensor configurations in terms of the
simplicial complex formulation. Let $\sigma_a=\{a,b,c,d\}$,
$\sigma_r=\{r,s,t,u\}$, and $\sigma_w=\{w,x,y,z\}$ be three generic
tetrahedra given in terms of their vertices in $\bR^3$.

\begin{example}
\label{ex:44dd}
\begin{enumerate}
\item First consider the sensor configuration comprising $\sigma_a$ and
$\sigma_r$ with $a=r$. In this case all extreme points have
coordinates equal to either $0$ or $1$. Only one vertex of each
tetrahedron can be assigned a $1$. This is the only restriction,
resulting in 10 (dominant)~extreme points.
\item Now consider the case where the three tetrahedra are attached in a
string:
$a=r$ and $u=w$.
As the same as the previous case, all extreme points have coordinates
equal to either $0$ or $1$.
\item Next, suppose that the three tetrahedra are connected in a
  cycle: $a=r$ and $u=w$, $x=b$. Then, in addition to the obvious
  $0,1$ extreme points,   there is one $0,\frac12$ extreme point with
  values $\frac12$ at $a$, $u$, and $x$.
\item Now suppose that the tetrahedra are connected along edges, say,
  $(a,b)=(r,s)$ and $(t,u)=(w,x)$. then again all extreme points are
  $0,1$-valued.
\item Finally consider the case $(a,b)=(r,s)$, $(t,u)=(w,x)$,
  $(y,z)=(c,d)$. In this case there also  $0,\frac12$-valued extreme
  points with supports triangles with one edge in each tetrahedron.
\end{enumerate}

Surprisingly, even if more and more tetrahedra are glued together in
such a way, the coordinates of all extreme points still only can be
$0$ or $1$, although the number of them increases quickly. We will
discuss these conclusions in much more detail in a later paper.
\end{example}

\section{Conclusion}

This paper introduces methods for obtaining the limits on fusion of
simple measurements from multiple sensors. The problem is
formulated as one in linear programming and dualized to obtain an
object, called the fusion polytope, which is the key device for
calculation of the minimum fused value compatible with the data. The
fusion polytope is computed in some simple cases. It is shown that
when the sensor configuration satisfies a simple property, defined in
the paper as \emph{genericity}, this fusion polytope can be assumed to
be in the positive orthant.

It is also shown that under some mild hypotheses genericity is not
satisfied when there are four overlapping sensor regions in the
plane. Consideration is given to the case of overlapping regions in a
hexagonal tiling in the plane. This is a generic situation, but simulations have shown
that it still leads to very complex extreme point structures.

We believe that the ideas described here are capable of being extended to more
complex situations where data fusion is required.

% Start of "Sample References" section

% Acknowledgments
\section*{Acknowledgment}
This work was supported by the US Defense Advanced Research Projects
Agency (DARPA) under grants Nos. 2006-06918-01--03,
by the US Air Force Office of Scientific Research (AFOSR) under grant No.
FA2386-13-1-4080, and by the  Institute for Mathematics and its
Applications (IMA).

% Bibliography
\bibliographystyle{IEEEtran}
\bibliography{counting_refs}
                             % Sample .bib file with references that match those in
                             % the 'Specifications Document (V1.5)' as well containing
                             % 'legacy' bibs and bibs with 'alternate codings'.
                             % Gerry Murray - March 2012

% History dates
%\received{February 2007}{March 2009}{June 2009}

\end{document}

%% file: frontmatter.tex
\usepackage{enumerate,graphicx,ifthen,color,wrapfig,comment}
\usepackage{amsmath,amscd,latexsym,amssymb,xspace}
\usepackage{microtype,todonotes,subfigure}
\usepackage[mathcal]{euscript}
\newcommand{\mnote}[1]{\ifthenelse{\isodd{\value{page}}}{\marginpar{{\color{red}{$\Longleftarrow$}}\parbox{1in}{\color{red}\small #1}}}{\marginpar{\parbox{1in}{\color{red}\small #1}{\color{red}{$\Longrightarrow$}}}}}

\newcommand{\loccit}{\emph{loc.\,cit.}\xspace}
\newcommand{\minval}{\mathfrak{m}}

\newcommand{\bd}{\ensuremath{{\mathbf d}}}
\newcommand{\ext}{\operatorname{{Ext}}}
\newcommand{\stab}{\operatorname{STAB}}
\newcommand{\qstab}{\operatorname{QSTAB}}
\newcommand{\fract}{\operatorname{FRAC}}
\newcommand{\co}{\operatorname{co}}
\newcommand{\bS}{\ensuremath{{\mathbf S}}}
\newcommand{\fA}{\ensuremath{{\frak A}}}

\newcommand\half{\ensuremath\frac{1}{2}}
\newcommand{\wy}{\widetilde{\by}}

\newcommand{\cV}{\ensuremath{{\mathcal V}}}
\newcommand{\cE}{\ensuremath{{\mathcal E}}}

\newcommand{\ba}{\ensuremath{{\mathbf a}}}
\newcommand{\bn}{\ensuremath{{\mathbf n}}}

\newcommand{\bc}{\ensuremath{{\mathbf c}}}

\newcommand{\bone}{\ensuremath{{\mathbf 1}}}
\newcommand{\cS}{\ensuremath{{\mathcal S}}}

\newcommand{\bR}{\ensuremath{{\mathbf R}}}
\newcommand{\bt}{\ensuremath{{\mathbf t}}}

\newcommand{\be}{\ensuremath{{\mathbf e}}}
\newcommand{\by}{\ensuremath{{\mathbf y}}}

\newcommand{\cG}{\ensuremath{{\mathcal G}}}
\newcommand{\cC}{\ensuremath{{\mathcal C}}}
\newcommand{\bx}{\ensuremath{{\mathbf x}}}

\newcommand{\bE}{\ensuremath{{\mathbf E}}}

%% file: white_31_regions.tex
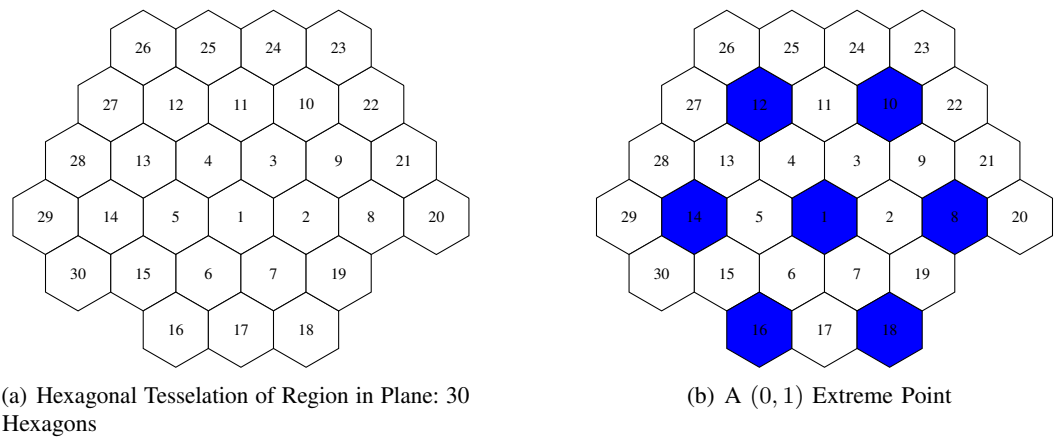
\begin{figure}[h!]
  \centering
  \subfigure[Hexagonal Tesselation of Region in Plane: 30 Hexagons]
{  \label{fig:31hex}
 \begin{tikzpicture}[scale=0.5]
 \tikzstyle{every node}=[font=\tiny]
 \draw[fill=
 white 
 ](
  .866025404,  .5
 )--(
  0.,  1.
 )--(
 -.866025404,  .5
 )--(
 -.866025404, -.5
 )--(
  0., -1.
 )--(
  .866025404, -.5
 )--(
  .866025404,  .5
 );
 \node at (
    .00000
 ,
    .00000
 ){ 1};
 \draw[fill=
 white 
 ](
  2.59807621,  .5
 )--(
  1.73205081,  1.
 )--(
  .866025404,  .5
 )--(
  .866025404, -.5
 )--(
  1.73205081, -1.
 )--(
  2.59807621, -.5
 )--(
  2.59807621,  .5
 );
 \node at (
   1.73205
 ,
    .00000
 ){ 2};
 \draw[fill=
 white 
 ](
  1.73205081,  2.
 )--(
  .866025404,  2.5
 )--(
  0.,  2.
 )--(
  0.,  1.
 )--(
  .866025404,  .5
 )--(
  1.73205081,  1.
 )--(
  1.73205081,  2.
 );
 \node at (
    .86603
 ,
   1.50000
 ){ 3};
 \draw[fill=
 white 
 ](
  0.,  2.
 )--(
 -.866025404,  2.5
 )--(
 -1.73205081,  2.
 )--(
 -1.73205081,  1.
 )--(
 -.866025404,  .5
 )--(
  0.,  1.
 )--(
  0.,  2.
 );
 \node at (
   -.86603
 ,
   1.50000
 ){ 4};
 \draw[fill=
 white 
 ](
 -.866025404,  .5
 )--(
 -1.73205081,  1.
 )--(
 -2.59807621,  .5
 )--(
 -2.59807621, -.5
 )--(
 -1.73205081, -1.
 )--(
 -.866025404, -.5
 )--(
 -.866025404,  .5
 );
 \node at (
  -1.73205
 ,
    .00000
 ){ 5};
 \draw[fill=
 white 
 ](
  0., -1.
 )--(
 -.866025404, -.5
 )--(
 -1.73205081, -1.
 )--(
 -1.73205081, -2.
 )--(
 -.866025404, -2.5
 )--(
  0., -2.
 )--(
  0., -1.
 );
 \node at (
   -.86603
 ,
  -1.50000
 ){ 6};
 \draw[fill=
 white 
 ](
  1.73205081, -1.
 )--(
  .866025404, -.5
 )--(
  0., -1.
 )--(
  0., -2.
 )--(
  .866025404, -2.5
 )--(
  1.73205081, -2.
 )--(
  1.73205081, -1.
 );
 \node at (
    .86603
 ,
  -1.50000
 ){ 7};
 \draw[fill=
 white 
 ](
  4.33012702,  .5
 )--(
  3.46410162,  1.
 )--(
  2.59807621,  .5
 )--(
  2.59807621, -.5
 )--(
  3.46410162, -1.
 )--(
  4.33012702, -.5
 )--(
  4.33012702,  .5
 );
 \node at (
   3.46410
 ,
    .00000
 ){ 8};
 \draw[fill=
 white 
 ](
  3.46410162,  2.
 )--(
  2.59807621,  2.5
 )--(
  1.73205081,  2.
 )--(
  1.73205081,  1.
 )--(
  2.59807621,  .5
 )--(
  3.46410162,  1.
 )--(
  3.46410162,  2.
 );
 \node at (
   2.59808
 ,
   1.50000
 ){ 9};
 \draw[fill=
 white 
 ](
  2.59807621,  3.5
 )--(
  1.73205081,  4.
 )--(
  .866025404,  3.5
 )--(
  .866025404,  2.5
 )--(
  1.73205081,  2.
 )--(
  2.59807621,  2.5
 )--(
  2.59807621,  3.5
 );
 \node at (
   1.73205
 ,
   3.00000
 ){ 10};
 \draw[fill=
 white 
 ](
  .866025404,  3.5
 )--(
  0.,  4.
 )--(
 -.866025404,  3.5
 )--(
 -.866025404,  2.5
 )--(
  0.,  2.
 )--(
  .866025404,  2.5
 )--(
  .866025404,  3.5
 );
 \node at (
    .00000
 ,
   3.00000
 ){ 11};
 \draw[fill=
 white 
 ](
 -.866025404,  3.5
 )--(
 -1.73205081,  4.
 )--(
 -2.59807621,  3.5
 )--(
 -2.59807621,  2.5
 )--(
 -1.73205081,  2.
 )--(
 -.866025404,  2.5
 )--(
 -.866025404,  3.5
 );
 \node at (
  -1.73205
 ,
   3.00000
 ){ 12};
 \draw[fill=
 white 
 ](
 -1.73205081,  2.
 )--(
 -2.59807621,  2.5
 )--(
 -3.46410162,  2.
 )--(
 -3.46410162,  1.
 )--(
 -2.59807621,  .5
 )--(
 -1.73205081,  1.
 )--(
 -1.73205081,  2.
 );
 \node at (
  -2.59808
 ,
   1.50000
 ){ 13};
 \draw[fill=
 white 
 ](
 -2.59807621,  .5
 )--(
 -3.46410162,  1.
 )--(
 -4.33012702,  .5
 )--(
 -4.33012702, -.5
 )--(
 -3.46410162, -1.
 )--(
 -2.59807621, -.5
 )--(
 -2.59807621,  .5
 );
 \node at (
  -3.46410
 ,
    .00000
 ){ 14};
 \draw[fill=
 white   
 ](
 -1.73205081, -1.
 )--(
 -2.59807621, -.5
 )--(
 -3.46410162, -1.
 )--(
 -3.46410162, -2.
 )--(
 -2.59807621, -2.5
 )--(
 -1.73205081, -2.
 )--(
 -1.73205081, -1.
 );
 \node at (
  -2.59808
 ,
  -1.50000
 ){ 15};
 \draw[fill=
 white 
 ](
 -.866025404, -2.5
 )--(
 -1.73205081, -2.
 )--(
 -2.59807621, -2.5
 )--(
 -2.59807621, -3.5
 )--(
 -1.73205081, -4.
 )--(
 -.866025404, -3.5
 )--(
 -.866025404, -2.5
 );
 \node at (
  -1.73205
 ,
  -3.00000
 ){ 16};
 \draw[fill=
 white   
 ](
  .866025404, -2.5
 )--(
  0., -2.
 )--(
 -.866025404, -2.5
 )--(
 -.866025404, -3.5
 )--(
  0., -4.
 )--(
  .866025404, -3.5
 )--(
  .866025404, -2.5
 );
 \node at (
    .00000
 ,
  -3.00000
 ){ 17};
 \draw[fill=
 white 
 ](
  2.59807621, -2.5
 )--(
  1.73205081, -2.
 )--(
  .866025404, -2.5
 )--(
  .866025404, -3.5
 )--(
  1.73205081, -4.
 )--(
  2.59807621, -3.5
 )--(
  2.59807621, -2.5
 );
 \node at (
   1.73205
 ,
  -3.00000
 ){ 18};
 \draw[fill=
 white   
 ](
  3.46410162, -1.
 )--(
  2.59807621, -.5
 )--(
  1.73205081, -1.
 )--(
  1.73205081, -2.
 )--(
  2.59807621, -2.5
 )--(
  3.46410162, -2.
 )--(
  3.46410162, -1.
 );
 \node at (
   2.59808
 ,
  -1.50000
 ){ 19};
 \draw[fill=
 white  
 ](
  6.06217783,  .5
 )--(
  5.19615242,  1.
 )--(
  4.33012702,  .5
 )--(
  4.33012702, -.5
 )--(
  5.19615242, -1.
 )--(
  6.06217783, -.5
 )--(
  6.06217783,  .5
 );
 \node at (
   5.19615
 ,
    .00000
 ){ 20};
 \draw[fill=
 white 
 ](
  5.19615242,  2.
 )--(
  4.33012702,  2.5
 )--(
  3.46410162,  2.
 )--(
  3.46410162,  1.
 )--(
  4.33012702,  .5
 )--(
  5.19615242,  1.
 )--(
  5.19615242,  2.
 );
 \node at (
   4.33013
 ,
   1.50000
 ){ 21};
 \draw[fill=
 white 
 ](
  4.33012702,  3.5
 )--(
  3.46410162,  4.
 )--(
  2.59807621,  3.5
 )--(
  2.59807621,  2.5
 )--(
  3.46410162,  2.
 )--(
  4.33012702,  2.5
 )--(
  4.33012702,  3.5
 );
 \node at (
   3.46410
 ,
   3.00000
 ){ 22};
 \draw[fill=
 white   
 ](
  3.46410162,  5.
 )--(
  2.59807621,  5.5
 )--(
  1.73205081,  5.
 )--(
  1.73205081,  4.
 )--(
  2.59807621,  3.5
 )--(
  3.46410162,  4.
 )--(
  3.46410162,  5.
 );
 \node at (
   2.59808
 ,
   4.50000
 ){ 23};
 \draw[fill=
 white 
 ](
  1.73205081,  5.
 )--(
  .866025404,  5.5
 )--(
  0.,  5.
 )--(
  0.,  4.
 )--(
  .866025404,  3.5
 )--(
  1.73205081,  4.
 )--(
  1.73205081,  5.
 );
 \node at (
    .86603
 ,
   4.50000
 ){ 24};
 \draw[fill=
 white   
 ](
  0.,  5.
 )--(
 -.866025404,  5.5
 )--(
 -1.73205081,  5.
 )--(
 -1.73205081,  4.
 )--(
 -.866025404,  3.5
 )--(
  0.,  4.
 )--(
  0.,  5.
 );
 \node at (
   -.86603
 ,
   4.50000
 ){ 25};
 \draw[fill=
 white 
 ](
 -1.73205081,  5.
 )--(
 -2.59807621,  5.5
 )--(
 -3.46410162,  5.
 )--(
 -3.46410162,  4.
 )--(
 -2.59807621,  3.5
 )--(
 -1.73205081,  4.
 )--(
 -1.73205081,  5.
 );
 \node at (
  -2.59808
 ,
   4.50000
 ){ 26};
 \draw[fill=
 white   
 ](
 -2.59807621,  3.5
 )--(
 -3.46410162,  4.
 )--(
 -4.33012702,  3.5
 )--(
 -4.33012702,  2.5
 )--(
 -3.46410162,  2.
 )--(
 -2.59807621,  2.5
 )--(
 -2.59807621,  3.5
 );
 \node at (
  -3.46410
 ,
   3.00000
 ){ 27};
 \draw[fill=
 white 
 ](
 -3.46410162,  2.
 )--(
 -4.33012702,  2.5
 )--(
 -5.19615242,  2.
 )--(
 -5.19615242,  1.
 )--(
 -4.33012702,  .5
 )--(
 -3.46410162,  1.
 )--(
 -3.46410162,  2.
 );
 \node at (
  -4.33013
 ,
   1.50000
 ){ 28};
 \draw[fill=
 white   
 ](
 -4.33012702,  .5
 )--(
 -5.19615242,  1.
 )--(
 -6.06217783,  .5
 )--(
 -6.06217783, -.5
 )--(
 -5.19615242, -1.
 )--(
 -4.33012702, -.5
 )--(
 -4.33012702,  .5
 );
 \node at (
  -5.19615
 ,
    .00000
 ){ 29};
 \draw[fill=
 white 
 ](
 -3.46410162, -1.
 )--(
 -4.33012702, -.5
 )--(
 -5.19615242, -1.
 )--(
 -5.19615242, -2.
 )--(
 -4.33012702, -2.5
 )--(
 -3.46410162, -2.
 )--(
 -3.46410162, -1.
 );
 \node at (
  -4.33013
 ,
  -1.50000
 ){ 30};
\end{tikzpicture}}\qquad \qquad
\subfigure[A $(0,1)$ Extreme Point]{
  \label{fig:white_blue_31hex}
  \begin{tikzpicture}[scale=0.5]
 \tikzstyle{every node}=[font=\tiny]
 \draw[fill=
 blue 
 ](
  .866025404,  .5
 )--(
  0.,  1.
 )--(
 -.866025404,  .5
 )--(
 -.866025404, -.5
 )--(
  0., -1.
 )--(
  .866025404, -.5
 )--(
  .866025404,  .5
 );
 \node at (
    .00000
 ,
    .00000
 ){1};
 \draw[fill=
 white 
 ](
  2.59807621,  .5
 )--(
  1.73205081,  1.
 )--(
  .866025404,  .5
 )--(
  .866025404, -.5
 )--(
  1.73205081, -1.
 )--(
  2.59807621, -.5
 )--(
  2.59807621,  .5
 );
 \node at (
   1.73205
 ,
    .00000
 ){ 2};
 \draw[fill=
 white 
 ](
  1.73205081,  2.
 )--(
  .866025404,  2.5
 )--(
  0.,  2.
 )--(
  0.,  1.
 )--(
  .866025404,  .5
 )--(
  1.73205081,  1.
 )--(
  1.73205081,  2.
 );
 \node at (
    .86603
 ,
   1.50000
 ){ 3};
 \draw[fill=
 white 
 ](
  0.,  2.
 )--(
 -.866025404,  2.5
 )--(
 -1.73205081,  2.
 )--(
 -1.73205081,  1.
 )--(
 -.866025404,  .5
 )--(
  0.,  1.
 )--(
  0.,  2.
 );
 \node at (
   -.86603
 ,
   1.50000
 ){ 4};
 \draw[fill=
 white 
 ](
 -.866025404,  .5
 )--(
 -1.73205081,  1.
 )--(
 -2.59807621,  .5
 )--(
 -2.59807621, -.5
 )--(
 -1.73205081, -1.
 )--(
 -.866025404, -.5
 )--(
 -.866025404,  .5
 );
 \node at (
  -1.73205
 ,
    .00000
 ){ 5};
 \draw[fill=
 white 
 ](
  0., -1.
 )--(
 -.866025404, -.5
 )--(
 -1.73205081, -1.
 )--(
 -1.73205081, -2.
 )--(
 -.866025404, -2.5
 )--(
  0., -2.
 )--(
  0., -1.
 );
 \node at (
   -.86603
 ,
  -1.50000
 ){ 6};
 \draw[fill=
 white 
 ](
  1.73205081, -1.
 )--(
  .866025404, -.5
 )--(
  0., -1.
 )--(
  0., -2.
 )--(
  .866025404, -2.5
 )--(
  1.73205081, -2.
 )--(
  1.73205081, -1.
 );
 \node at (
    .86603
 ,
  -1.50000
 ){ 7};
 \draw[fill=
 blue 
 ](
  4.33012702,  .5
 )--(
  3.46410162,  1.
 )--(
  2.59807621,  .5
 )--(
  2.59807621, -.5
 )--(
  3.46410162, -1.
 )--(
  4.33012702, -.5
 )--(
  4.33012702,  .5
 );
 \node at (
   3.46410
 ,
    .00000
 ){ 8};
 \draw[fill=
 white 
 ](
  3.46410162,  2.
 )--(
  2.59807621,  2.5
 )--(
  1.73205081,  2.
 )--(
  1.73205081,  1.
 )--(
  2.59807621,  .5
 )--(
  3.46410162,  1.
 )--(
  3.46410162,  2.
 );
 \node at (
   2.59808
 ,
   1.50000
 ){ 9};
 \draw[fill=
 blue
 ](
  2.59807621,  3.5
 )--(
  1.73205081,  4.
 )--(
  .866025404,  3.5
 )--(
  .866025404,  2.5
 )--(
  1.73205081,  2.
 )--(
  2.59807621,  2.5
 )--(
  2.59807621,  3.5
 );
 \node at (
   1.73205
 ,
   3.00000
 ){ 10};
 \draw[fill=
 white
 ](
  .866025404,  3.5
 )--(
  0.,  4.
 )--(
 -.866025404,  3.5
 )--(
 -.866025404,  2.5
 )--(
  0.,  2.
 )--(
  .866025404,  2.5
 )--(
  .866025404,  3.5
 );
 \node at (
    .00000
 ,
   3.00000
 ){ 11};
 \draw[fill=
 blue 
 ](
 -.866025404,  3.5
 )--(
 -1.73205081,  4.
 )--(
 -2.59807621,  3.5
 )--(
 -2.59807621,  2.5
 )--(
 -1.73205081,  2.
 )--(
 -.866025404,  2.5
 )--(
 -.866025404,  3.5
 );
 \node at (
  -1.73205
 ,
   3.00000
 ){ 12};
 \draw[fill=
 white 
 ](
 -1.73205081,  2.
 )--(
 -2.59807621,  2.5
 )--(
 -3.46410162,  2.
 )--(
 -3.46410162,  1.
 )--(
 -2.59807621,  .5
 )--(
 -1.73205081,  1.
 )--(
 -1.73205081,  2.
 );
 \node at (
  -2.59808
 ,
   1.50000
 ){ 13};
 \draw[fill=
 blue 
 ](
 -2.59807621,  .5
 )--(
 -3.46410162,  1.
 )--(
 -4.33012702,  .5
 )--(
 -4.33012702, -.5
 )--(
 -3.46410162, -1.
 )--(
 -2.59807621, -.5
 )--(
 -2.59807621,  .5
 );
 \node at (
  -3.46410
 ,
    .00000
 ){ 14};
 \draw[fill=
 white   
 ](
 -1.73205081, -1.
 )--(
 -2.59807621, -.5
 )--(
 -3.46410162, -1.
 )--(
 -3.46410162, -2.
 )--(
 -2.59807621, -2.5
 )--(
 -1.73205081, -2.
 )--(
 -1.73205081, -1.
 );
 \node at (
  -2.59808
 ,
  -1.50000
 ){ 15};
 \draw[fill=
 blue
 ](
 -.866025404, -2.5
 )--(
 -1.73205081, -2.
 )--(
 -2.59807621, -2.5
 )--(
 -2.59807621, -3.5
 )--(
 -1.73205081, -4.
 )--(
 -.866025404, -3.5
 )--(
 -.866025404, -2.5
 );
 \node at (
  -1.73205
 ,
  -3.00000
 ){ 16};
 \draw[fill=
 white   
 ](
  .866025404, -2.5
 )--(
  0., -2.
 )--(
 -.866025404, -2.5
 )--(
 -.866025404, -3.5
 )--(
  0., -4.
 )--(
  .866025404, -3.5
 )--(
  .866025404, -2.5
 );
 \node at (
    .00000
 ,
  -3.00000
 ){ 17};
 \draw[fill=
 blue 
 ](
  2.59807621, -2.5
 )--(
  1.73205081, -2.
 )--(
  .866025404, -2.5
 )--(
  .866025404, -3.5
 )--(
  1.73205081, -4.
 )--(
  2.59807621, -3.5
 )--(
  2.59807621, -2.5
 );
 \node at (
   1.73205
 ,
  -3.00000
 ){ 18};
 \draw[fill=
 white   
 ](
  3.46410162, -1.
 )--(
  2.59807621, -.5
 )--(
  1.73205081, -1.
 )--(
  1.73205081, -2.
 )--(
  2.59807621, -2.5
 )--(
  3.46410162, -2.
 )--(
  3.46410162, -1.
 );
 \node at (
   2.59808
 ,
  -1.50000
 ){ 19};
 \draw[fill=
 white  
 ](
  6.06217783,  .5
 )--(
  5.19615242,  1.
 )--(
  4.33012702,  .5
 )--(
  4.33012702, -.5
 )--(
  5.19615242, -1.
 )--(
  6.06217783, -.5
 )--(
  6.06217783,  .5
 );
 \node at (
   5.19615
 ,
    .00000
 ){ 20};
 \draw[fill=
 white 
 ](
  5.19615242,  2.
 )--(
  4.33012702,  2.5
 )--(
  3.46410162,  2.
 )--(
  3.46410162,  1.
 )--(
  4.33012702,  .5
 )--(
  5.19615242,  1.
 )--(
  5.19615242,  2.
 );
 \node at (
   4.33013
 ,
   1.50000
 ){ 21};
 \draw[fill=
 white 
 ](
  4.33012702,  3.5
 )--(
  3.46410162,  4.
 )--(
  2.59807621,  3.5
 )--(
  2.59807621,  2.5
 )--(
  3.46410162,  2.
 )--(
  4.33012702,  2.5
 )--(
  4.33012702,  3.5
 );
 \node at (
   3.46410
 ,
   3.00000
 ){ 22};
 \draw[fill=
 white   
 ](
  3.46410162,  5.
 )--(
  2.59807621,  5.5
 )--(
  1.73205081,  5.
 )--(
  1.73205081,  4.
 )--(
  2.59807621,  3.5
 )--(
  3.46410162,  4.
 )--(
  3.46410162,  5.
 );
 \node at (
   2.59808
 ,
   4.50000
 ){ 23};
 \draw[fill=
 white 
 ](
  1.73205081,  5.
 )--(
  .866025404,  5.5
 )--(
  0.,  5.
 )--(
  0.,  4.
 )--(
  .866025404,  3.5
 )--(
  1.73205081,  4.
 )--(
  1.73205081,  5.
 );
 \node at (
    .86603
 ,
   4.50000
 ){ 24};
 \draw[fill=
 white   
 ](
  0.,  5.
 )--(
 -.866025404,  5.5
 )--(
 -1.73205081,  5.
 )--(
 -1.73205081,  4.
 )--(
 -.866025404,  3.5
 )--(
  0.,  4.
 )--(
  0.,  5.
 );
 \node at (
   -.86603
 ,
   4.50000
 ){ 25};
 \draw[fill=
 white 
 ](
 -1.73205081,  5.
 )--(
 -2.59807621,  5.5
 )--(
 -3.46410162,  5.
 )--(
 -3.46410162,  4.
 )--(
 -2.59807621,  3.5
 )--(
 -1.73205081,  4.
 )--(
 -1.73205081,  5.
 );
 \node at (
  -2.59808
 ,
   4.50000
 ){ 26};
 \draw[fill=
 white   
 ](
 -2.59807621,  3.5
 )--(
 -3.46410162,  4.
 )--(
 -4.33012702,  3.5
 )--(
 -4.33012702,  2.5
 )--(
 -3.46410162,  2.
 )--(
 -2.59807621,  2.5
 )--(
 -2.59807621,  3.5
 );
 \node at (
  -3.46410
 ,
   3.00000
 ){ 27};
 \draw[fill=
 white 
 ](
 -3.46410162,  2.
 )--(
 -4.33012702,  2.5
 )--(
 -5.19615242,  2.
 )--(
 -5.19615242,  1.
 )--(
 -4.33012702,  .5
 )--(
 -3.46410162,  1.
 )--(
 -3.46410162,  2.
 );
 \node at (
  -4.33013
 ,
   1.50000
 ){ 28};
 \draw[fill=
 white   
 ](
 -4.33012702,  .5
 )--(
 -5.19615242,  1.
 )--(
 -6.06217783,  .5
 )--(
 -6.06217783, -.5
 )--(
 -5.19615242, -1.
 )--(
 -4.33012702, -.5
 )--(
 -4.33012702,  .5
 );
 \node at (
  -5.19615
 ,
    .00000
 ){ 29};
 \draw[fill=
 white 
 ](
 -3.46410162, -1.
 )--(
 -4.33012702, -.5
 )--(
 -5.19615242, -1.
 )--(
 -5.19615242, -2.
 )--(
 -4.33012702, -2.5
 )--(
 -3.46410162, -2.
 )--(
 -3.46410162, -1.
 );
 \node at (
  -4.33013
 ,
  -1.50000
 ){ 30};
\end{tikzpicture}}
\caption{Tesselation and Extreme point involving only $0,1$ coefficients}
\end{figure}

%%% Local Variables: 
%%% mode: latex
%%% TeX-master: "graphs"
%%% End: 

%% file: white_yellow_31_regions.tex
\begin{figure}[h!]
  \centering
\subfigure[Extreme point with $0, \frac{1}{2}$ coefficients]{ \begin{tikzpicture}[scale=0.5]
 \tikzstyle{every node}=[font=\tiny]
 \draw[fill=
 white 
 ](
  .866025404,  .5
 )--(
  0.,  1.
 )--(
 -.866025404,  .5
 )--(
 -.866025404, -.5
 )--(
  0., -1.
 )--(
  .866025404, -.5
 )--(
  .866025404,  .5
 );
 \node at (
    .00000
 ,
    .00000
 ){ 1};
 \draw[fill=
 white 
 ](
  2.59807621,  .5
 )--(
  1.73205081,  1.
 )--(
  .866025404,  .5
 )--(
  .866025404, -.5
 )--(
  1.73205081, -1.
 )--(
  2.59807621, -.5
 )--(
  2.59807621,  .5
 );
 \node at (
   1.73205
 ,
    .00000
 ){ 2};
 \draw[fill=
 yellow 
 ](
  1.73205081,  2.
 )--(
  .866025404,  2.5
 )--(
  0.,  2.
 )--(
  0.,  1.
 )--(
  .866025404,  .5
 )--(
  1.73205081,  1.
 )--(
  1.73205081,  2.
 );
 \node at (
    .86603
 ,
   1.50000
 ){ 3};
 \draw[fill=
 yellow 
 ](
  0.,  2.
 )--(
 -.866025404,  2.5
 )--(
 -1.73205081,  2.
 )--(
 -1.73205081,  1.
 )--(
 -.866025404,  .5
 )--(
  0.,  1.
 )--(
  0.,  2.
 );
 \node at (
   -.86603
 ,
   1.50000
 ){ 4};
 \draw[fill=
 yellow 
 ](
 -.866025404,  .5
 )--(
 -1.73205081,  1.
 )--(
 -2.59807621,  .5
 )--(
 -2.59807621, -.5
 )--(
 -1.73205081, -1.
 )--(
 -.866025404, -.5
 )--(
 -.866025404,  .5
 );
 \node at (
  -1.73205
 ,
    .00000
 ){ 5};
 \draw[fill=
 yellow  
 ](
  0., -1.
 )--(
 -.866025404, -.5
 )--(
 -1.73205081, -1.
 )--(
 -1.73205081, -2.
 )--(
 -.866025404, -2.5
 )--(
  0., -2.
 )--(
  0., -1.
 );
 \node at (
   -.86603
 ,
  -1.50000
 ){ 6};
 \draw[fill=
 white 
 ](
  1.73205081, -1.
 )--(
  .866025404, -.5
 )--(
  0., -1.
 )--(
  0., -2.
 )--(
  .866025404, -2.5
 )--(
  1.73205081, -2.
 )--(
  1.73205081, -1.
 );
 \node at (
    .86603
 ,
  -1.50000
 ){ 7};
 \draw[fill=
 yellow
 ](
  4.33012702,  .5
 )--(
  3.46410162,  1.
 )--(
  2.59807621,  .5
 )--(
  2.59807621, -.5
 )--(
  3.46410162, -1.
 )--(
  4.33012702, -.5
 )--(
  4.33012702,  .5
 );
 \node at (
   3.46410
 ,
    .00000
 ){ 8};
 \draw[fill=
 white 
 ](
  3.46410162,  2.
 )--(
  2.59807621,  2.5
 )--(
  1.73205081,  2.
 )--(
  1.73205081,  1.
 )--(
  2.59807621,  .5
 )--(
  3.46410162,  1.
 )--(
  3.46410162,  2.
 );
 \node at (
   2.59808
 ,
   1.50000
 ){ 9};
 \draw[fill=
 yellow 
 ](
  2.59807621,  3.5
 )--(
  1.73205081,  4.
 )--(
  .866025404,  3.5
 )--(
  .866025404,  2.5
 )--(
  1.73205081,  2.
 )--(
  2.59807621,  2.5
 )--(
  2.59807621,  3.5
 );
 \node at (
   1.73205
 ,
   3.00000
 ){ 10};
 \draw[fill=
 white 
 ](
  .866025404,  3.5
 )--(
  0.,  4.
 )--(
 -.866025404,  3.5
 )--(
 -.866025404,  2.5
 )--(
  0.,  2.
 )--(
  .866025404,  2.5
 )--(
  .866025404,  3.5
 );
 \node at (
    .00000
 ,
   3.00000
 ){ 11};
 \draw[fill=
 yellow 
 ](
 -.866025404,  3.5
 )--(
 -1.73205081,  4.
 )--(
 -2.59807621,  3.5
 )--(
 -2.59807621,  2.5
 )--(
 -1.73205081,  2.
 )--(
 -.866025404,  2.5
 )--(
 -.866025404,  3.5
 );
 \node at (
  -1.73205
 ,
   3.00000
 ){ 12};
 \draw[fill=
 white 
 ](
 -1.73205081,  2.
 )--(
 -2.59807621,  2.5
 )--(
 -3.46410162,  2.
 )--(
 -3.46410162,  1.
 )--(
 -2.59807621,  .5
 )--(
 -1.73205081,  1.
 )--(
 -1.73205081,  2.
 );
 \node at (
  -2.59808
 ,
   1.50000
 ){ 13};
 \draw[fill=
 yellow  
 ](
 -2.59807621,  .5
 )--(
 -3.46410162,  1.
 )--(
 -4.33012702,  .5
 )--(
 -4.33012702, -.5
 )--(
 -3.46410162, -1.
 )--(
 -2.59807621, -.5
 )--(
 -2.59807621,  .5
 );
 \node at (
  -3.46410
 ,
    .00000
 ){ 14};
 \draw[fill=
 white   
 ](
 -1.73205081, -1.
 )--(
 -2.59807621, -.5
 )--(
 -3.46410162, -1.
 )--(
 -3.46410162, -2.
 )--(
 -2.59807621, -2.5
 )--(
 -1.73205081, -2.
 )--(
 -1.73205081, -1.
 );
 \node at (
  -2.59808
 ,
  -1.50000
 ){ 15};
 \draw[fill=
 white 
 ](
 -.866025404, -2.5
 )--(
 -1.73205081, -2.
 )--(
 -2.59807621, -2.5
 )--(
 -2.59807621, -3.5
 )--(
 -1.73205081, -4.
 )--(
 -.866025404, -3.5
 )--(
 -.866025404, -2.5
 );
 \node at (
  -1.73205
 ,
  -3.00000
 ){ 16};
 \draw[fill=
 yellow   
 ](
  .866025404, -2.5
 )--(
  0., -2.
 )--(
 -.866025404, -2.5
 )--(
 -.866025404, -3.5
 )--(
  0., -4.
 )--(
  .866025404, -3.5
 )--(
  .866025404, -2.5
 );
 \node at (
    .00000
 ,
  -3.00000
 ){ 17};
 \draw[fill=
 yellow 
 ](
  2.59807621, -2.5
 )--(
  1.73205081, -2.
 )--(
  .866025404, -2.5
 )--(
  .866025404, -3.5
 )--(
  1.73205081, -4.
 )--(
  2.59807621, -3.5
 )--(
  2.59807621, -2.5
 );
 \node at (
   1.73205
 ,
  -3.00000
 ){ 18};
 \draw[fill=
 yellow  
 ](
  3.46410162, -1.
 )--(
  2.59807621, -.5
 )--(
  1.73205081, -1.
 )--(
  1.73205081, -2.
 )--(
  2.59807621, -2.5
 )--(
  3.46410162, -2.
 )--(
  3.46410162, -1.
 );
 \node at (
   2.59808
 ,
  -1.50000
 ){ 19};
 \draw[fill=
 white  
 ](
  6.06217783,  .5
 )--(
  5.19615242,  1.
 )--(
  4.33012702,  .5
 )--(
  4.33012702, -.5
 )--(
  5.19615242, -1.
 )--(
  6.06217783, -.5
 )--(
  6.06217783,  .5
 );
 \node at (
   5.19615
 ,
    .00000
 ){ 20};
 \draw[fill=
 yellow
 ](
  5.19615242,  2.
 )--(
  4.33012702,  2.5
 )--(
  3.46410162,  2.
 )--(
  3.46410162,  1.
 )--(
  4.33012702,  .5
 )--(
  5.19615242,  1.
 )--(
  5.19615242,  2.
 );
 \node at (
   4.33013
 ,
   1.50000
 ){ 21};
 \draw[fill=
 yellow 
 ](
  4.33012702,  3.5
 )--(
  3.46410162,  4.
 )--(
  2.59807621,  3.5
 )--(
  2.59807621,  2.5
 )--(
  3.46410162,  2.
 )--(
  4.33012702,  2.5
 )--(
  4.33012702,  3.5
 );
 \node at (
   3.46410
 ,
   3.00000
 ){ 22};
 \draw[fill=
 white   
 ](
  3.46410162,  5.
 )--(
  2.59807621,  5.5
 )--(
  1.73205081,  5.
 )--(
  1.73205081,  4.
 )--(
  2.59807621,  3.5
 )--(
  3.46410162,  4.
 )--(
  3.46410162,  5.
 );
 \node at (
   2.59808
 ,
   4.50000
 ){ 23};
 \draw[fill=
 yellow 
 ](
  1.73205081,  5.
 )--(
  .866025404,  5.5
 )--(
  0.,  5.
 )--(
  0.,  4.
 )--(
  .866025404,  3.5
 )--(
  1.73205081,  4.
 )--(
  1.73205081,  5.
 );
 \node at (
    .86603
 ,
   4.50000
 ){ 24};
 \draw[fill=
 yellow    
 ](
  0.,  5.
 )--(
 -.866025404,  5.5
 )--(
 -1.73205081,  5.
 )--(
 -1.73205081,  4.
 )--(
 -.866025404,  3.5
 )--(
  0.,  4.
 )--(
  0.,  5.
 );
 \node at (
   -.86603
 ,
   4.50000
 ){ 25};
 \draw[fill=
 white 
 ](
 -1.73205081,  5.
 )--(
 -2.59807621,  5.5
 )--(
 -3.46410162,  5.
 )--(
 -3.46410162,  4.
 )--(
 -2.59807621,  3.5
 )--(
 -1.73205081,  4.
 )--(
 -1.73205081,  5.
 );
 \node at (
  -2.59808
 ,
   4.50000
 ){ 26};
 \draw[fill=
 yellow    
 ](
 -2.59807621,  3.5
 )--(
 -3.46410162,  4.
 )--(
 -4.33012702,  3.5
 )--(
 -4.33012702,  2.5
 )--(
 -3.46410162,  2.
 )--(
 -2.59807621,  2.5
 )--(
 -2.59807621,  3.5
 );
 \node at (
  -3.46410
 ,
   3.00000
 ){ 27};
 \draw[fill=
 yellow  
 ](
 -3.46410162,  2.
 )--(
 -4.33012702,  2.5
 )--(
 -5.19615242,  2.
 )--(
 -5.19615242,  1.
 )--(
 -4.33012702,  .5
 )--(
 -3.46410162,  1.
 )--(
 -3.46410162,  2.
 );
 \node at (
  -4.33013
 ,
   1.50000
 ){ 28};
 \draw[fill=
 white   
 ](
 -4.33012702,  .5
 )--(
 -5.19615242,  1.
 )--(
 -6.06217783,  .5
 )--(
 -6.06217783, -.5
 )--(
 -5.19615242, -1.
 )--(
 -4.33012702, -.5
 )--(
 -4.33012702,  .5
 );
 \node at (
  -5.19615
 ,
    .00000
 ){ 29};
 \draw[fill=
 yellow 
 ](
 -3.46410162, -1.
 )--(
 -4.33012702, -.5
 )--(
 -5.19615242, -1.
 )--(
 -5.19615242, -2.
 )--(
 -4.33012702, -2.5
 )--(
 -3.46410162, -2.
 )--(
 -3.46410162, -1.
 );
 \node at (
  -4.33013
 ,
  -1.50000
 ){ 30};
 \end{tikzpicture}
  \label{fig:white_yellow_31hex}}\qquad \qquad
\subfigure[Mixture of $0,1$ and $0, \frac{1}{2}$ extreme point]{
 \begin{tikzpicture}[scale=0.5]
 \tikzstyle{every node}=[font=\tiny]
 \draw[fill=
 yellow 
 ](
  .866025404,  .5
 )--(
  0.,  1.
 )--(
 -.866025404,  .5
 )--(
 -.866025404, -.5
 )--(
  0., -1.
 )--(
  .866025404, -.5
 )--(
  .866025404,  .5
 );
 \node at (
    .00000
 ,
    .00000
 ){ 1};
 \draw[fill=
 yellow 
 ](
  2.59807621,  .5
 )--(
  1.73205081,  1.
 )--(
  .866025404,  .5
 )--(
  .866025404, -.5
 )--(
  1.73205081, -1.
 )--(
  2.59807621, -.5
 )--(
  2.59807621,  .5
 );
 \node at (
   1.73205
 ,
    .00000
 ){ 2};
 \draw[fill=
 white 
 ](
  1.73205081,  2.
 )--(
  .866025404,  2.5
 )--(
  0.,  2.
 )--(
  0.,  1.
 )--(
  .866025404,  .5
 )--(
  1.73205081,  1.
 )--(
  1.73205081,  2.
 );
 \node at (
    .86603
 ,
   1.50000
 ){ 3};
 \draw[fill=
 yellow 
 ](
  0.,  2.
 )--(
 -.866025404,  2.5
 )--(
 -1.73205081,  2.
 )--(
 -1.73205081,  1.
 )--(
 -.866025404,  .5
 )--(
  0.,  1.
 )--(
  0.,  2.
 );
 \node at (
   -.86603
 ,
   1.50000
 ){ 4};
 \draw[fill=
 white 
 ](
 -.866025404,  .5
 )--(
 -1.73205081,  1.
 )--(
 -2.59807621,  .5
 )--(
 -2.59807621, -.5
 )--(
 -1.73205081, -1.
 )--(
 -.866025404, -.5
 )--(
 -.866025404,  .5
 );
 \node at (
  -1.73205
 ,
    .00000
 ){ 5};
 \draw[fill=
 white 
 ](
  0., -1.
 )--(
 -.866025404, -.5
 )--(
 -1.73205081, -1.
 )--(
 -1.73205081, -2.
 )--(
 -.866025404, -2.5
 )--(
  0., -2.
 )--(
  0., -1.
 );
 \node at (
   -.86603
 ,
  -1.50000
 ){ 6};
 \draw[fill=
 white 
 ](
  1.73205081, -1.
 )--(
  .866025404, -.5
 )--(
  0., -1.
 )--(
  0., -2.
 )--(
  .866025404, -2.5
 )--(
  1.73205081, -2.
 )--(
  1.73205081, -1.
 );
 \node at (
    .86603
 ,
  -1.50000
 ){ 7};
 \draw[fill=
 yellow 
 ](
  4.33012702,  .5
 )--(
  3.46410162,  1.
 )--(
  2.59807621,  .5
 )--(
  2.59807621, -.5
 )--(
  3.46410162, -1.
 )--(
  4.33012702, -.5
 )--(
  4.33012702,  .5
 );
 \node at (
   3.46410
 ,
    .00000
 ){ 8};
 \draw[fill=
 white 
 ](
  3.46410162,  2.
 )--(
  2.59807621,  2.5
 )--(
  1.73205081,  2.
 )--(
  1.73205081,  1.
 )--(
  2.59807621,  .5
 )--(
  3.46410162,  1.
 )--(
  3.46410162,  2.
 );
 \node at (
   2.59808
 ,
   1.50000
 ){ 9};
 \draw[fill=
 white 
 ](
  2.59807621,  3.5
 )--(
  1.73205081,  4.
 )--(
  .866025404,  3.5
 )--(
  .866025404,  2.5
 )--(
  1.73205081,  2.
 )--(
  2.59807621,  2.5
 )--(
  2.59807621,  3.5
 );
 \node at (
   1.73205
 ,
   3.00000
 ){ 10};
 \draw[fill=
 yellow 
 ](
  .866025404,  3.5
 )--(
  0.,  4.
 )--(
 -.866025404,  3.5
 )--(
 -.866025404,  2.5
 )--(
  0.,  2.
 )--(
  .866025404,  2.5
 )--(
  .866025404,  3.5
 );
 \node at (
    .00000
 ,
   3.00000
 ){ 11};
 \draw[fill=
 white 
 ](
 -.866025404,  3.5
 )--(
 -1.73205081,  4.
 )--(
 -2.59807621,  3.5
 )--(
 -2.59807621,  2.5
 )--(
 -1.73205081,  2.
 )--(
 -.866025404,  2.5
 )--(
 -.866025404,  3.5
 );
 \node at (
  -1.73205
 ,
   3.00000
 ){ 12};
 \draw[fill=
 white 
 ](
 -1.73205081,  2.
 )--(
 -2.59807621,  2.5
 )--(
 -3.46410162,  2.
 )--(
 -3.46410162,  1.
 )--(
 -2.59807621,  .5
 )--(
 -1.73205081,  1.
 )--(
 -1.73205081,  2.
 );
 \node at (
  -2.59808
 ,
   1.50000
 ){ 13};
 \draw[fill=
 white 
 ](
 -2.59807621,  .5
 )--(
 -3.46410162,  1.
 )--(
 -4.33012702,  .5
 )--(
 -4.33012702, -.5
 )--(
 -3.46410162, -1.
 )--(
 -2.59807621, -.5
 )--(
 -2.59807621,  .5
 );
 \node at (
  -3.46410
 ,
    .00000
 ){ 14};
 \draw[fill=
 white   
 ](
 -1.73205081, -1.
 )--(
 -2.59807621, -.5
 )--(
 -3.46410162, -1.
 )--(
 -3.46410162, -2.
 )--(
 -2.59807621, -2.5
 )--(
 -1.73205081, -2.
 )--(
 -1.73205081, -1.
 );
 \node at (
  -2.59808
 ,
  -1.50000
 ){ 15};
 \draw[fill=
 white 
 ](
 -.866025404, -2.5
 )--(
 -1.73205081, -2.
 )--(
 -2.59807621, -2.5
 )--(
 -2.59807621, -3.5
 )--(
 -1.73205081, -4.
 )--(
 -.866025404, -3.5
 )--(
 -.866025404, -2.5
 );
 \node at (
  -1.73205
 ,
  -3.00000
 ){ 16};
 \draw[fill=
 blue   
 ](
  .866025404, -2.5
 )--(
  0., -2.
 )--(
 -.866025404, -2.5
 )--(
 -.866025404, -3.5
 )--(
  0., -4.
 )--(
  .866025404, -3.5
 )--(
  .866025404, -2.5
 );
 \node at (
    .00000
 ,
  -3.00000
 ){ 17};
 \draw[fill=
 white 
 ](
  2.59807621, -2.5
 )--(
  1.73205081, -2.
 )--(
  .866025404, -2.5
 )--(
  .866025404, -3.5
 )--(
  1.73205081, -4.
 )--(
  2.59807621, -3.5
 )--(
  2.59807621, -2.5
 );
 \node at (
   1.73205
 ,
  -3.00000
 ){ 18};
 \draw[fill=
 white   
 ](
  3.46410162, -1.
 )--(
  2.59807621, -.5
 )--(
  1.73205081, -1.
 )--(
  1.73205081, -2.
 )--(
  2.59807621, -2.5
 )--(
  3.46410162, -2.
 )--(
  3.46410162, -1.
 );
 \node at (
   2.59808
 ,
  -1.50000
 ){ 19};
 \draw[fill=
 white  
 ](
  6.06217783,  .5
 )--(
  5.19615242,  1.
 )--(
  4.33012702,  .5
 )--(
  4.33012702, -.5
 )--(
  5.19615242, -1.
 )--(
  6.06217783, -.5
 )--(
  6.06217783,  .5
 );
 \node at (
   5.19615
 ,
    .00000
 ){ 20};
 \draw[fill=
 yellow
 ](
  5.19615242,  2.
 )--(
  4.33012702,  2.5
 )--(
  3.46410162,  2.
 )--(
  3.46410162,  1.
 )--(
  4.33012702,  .5
 )--(
  5.19615242,  1.
 )--(
  5.19615242,  2.
 );
 \node at (
   4.33013
 ,
   1.50000
 ){ 21};
 \draw[fill=
 yellow
 ](
  4.33012702,  3.5
 )--(
  3.46410162,  4.
 )--(
  2.59807621,  3.5
 )--(
  2.59807621,  2.5
 )--(
  3.46410162,  2.
 )--(
  4.33012702,  2.5
 )--(
  4.33012702,  3.5
 );
 \node at (
   3.46410
 ,
   3.00000
 ){ 22};
 \draw[fill=
 yellow   
 ](
  3.46410162,  5.
 )--(
  2.59807621,  5.5
 )--(
  1.73205081,  5.
 )--(
  1.73205081,  4.
 )--(
  2.59807621,  3.5
 )--(
  3.46410162,  4.
 )--(
  3.46410162,  5.
 );
 \node at (
   2.59808
 ,
   4.50000
 ){ 23};
 \draw[fill=
 yellow 
 ](
  1.73205081,  5.
 )--(
  .866025404,  5.5
 )--(
  0.,  5.
 )--(
  0.,  4.
 )--(
  .866025404,  3.5
 )--(
  1.73205081,  4.
 )--(
  1.73205081,  5.
 );
 \node at (
    .86603
 ,
   4.50000
 ){ 24};
 \draw[fill=
 white   
 ](
  0.,  5.
 )--(
 -.866025404,  5.5
 )--(
 -1.73205081,  5.
 )--(
 -1.73205081,  4.
 )--(
 -.866025404,  3.5
 )--(
  0.,  4.
 )--(
  0.,  5.
 );
 \node at (
   -.86603
 ,
   4.50000
 ){ 25};
 \draw[fill=
 blue 
 ](
 -1.73205081,  5.
 )--(
 -2.59807621,  5.5
 )--(
 -3.46410162,  5.
 )--(
 -3.46410162,  4.
 )--(
 -2.59807621,  3.5
 )--(
 -1.73205081,  4.
 )--(
 -1.73205081,  5.
 );
 \node at (
  -2.59808
 ,
   4.50000
 ){ 26};
 \draw[fill=
 white   
 ](
 -2.59807621,  3.5
 )--(
 -3.46410162,  4.
 )--(
 -4.33012702,  3.5
 )--(
 -4.33012702,  2.5
 )--(
 -3.46410162,  2.
 )--(
 -2.59807621,  2.5
 )--(
 -2.59807621,  3.5
 );
 \node at (
  -3.46410
 ,
   3.00000
 ){ 27};
 \draw[fill=
 blue 
 ](
 -3.46410162,  2.
 )--(
 -4.33012702,  2.5
 )--(
 -5.19615242,  2.
 )--(
 -5.19615242,  1.
 )--(
 -4.33012702,  .5
 )--(
 -3.46410162,  1.
 )--(
 -3.46410162,  2.
 );
 \node at (
  -4.33013
 ,
   1.50000
 ){ 28};
 \draw[fill=
 white   
 ](
 -4.33012702,  .5
 )--(
 -5.19615242,  1.
 )--(
 -6.06217783,  .5
 )--(
 -6.06217783, -.5
 )--(
 -5.19615242, -1.
 )--(
 -4.33012702, -.5
 )--(
 -4.33012702,  .5
 );
 \node at (
  -5.19615
 ,
    .00000
 ){ 29};
 \draw[fill=
 blue 
 ](
 -3.46410162, -1.
 )--(
 -4.33012702, -.5
 )--(
 -5.19615242, -1.
 )--(
 -5.19615242, -2.
 )--(
 -4.33012702, -2.5
 )--(
 -3.46410162, -2.
 )--(
 -3.46410162, -1.
 );
 \node at (
  -4.33013
 ,
  -1.50000
 ){ 30};
 \end{tikzpicture}
  \label{fig:whiteblueyellow_hex}
}

\caption{Extreme points involving $0, \frac{1}{2}, 1$ coefficients}
\end{figure}
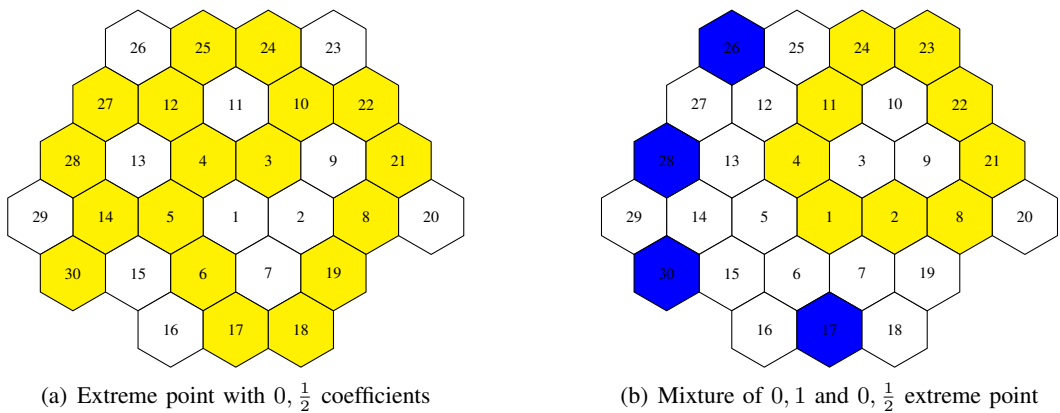

%%% Local Variables: 
%%% mode: latex
%%% TeX-master: "graphs"
%%% End: 

%% file: greyblackyellow_30regions.tex
\begin{figure}[h!]
  \centering

\subfigure[An extreme point with coefficients $0,\frac{1}{2},
    \frac{1}{4}, \frac{3}{4}$]{
  \begin{tikzpicture}[scale=0.5]
 \tikzstyle{every node}=[font=\tiny]
 \draw[fill=
 gray 
 ](
  .866025404,  .5
 )--(
  0.,  1.
 )--(
 -.866025404,  .5
 )--(
 -.866025404, -.5
 )--(
  0., -1.
 )--(
  .866025404, -.5
 )--(
  .866025404,  .5
 );
 \node at (
    .00000
 ,
    .00000
 ){ 1};
 \draw[fill=
 white 
 ](
  2.59807621,  .5
 )--(
  1.73205081,  1.
 )--(
  .866025404,  .5
 )--(
  .866025404, -.5
 )--(
  1.73205081, -1.
 )--(
  2.59807621, -.5
 )--(
  2.59807621,  .5
 );
 \node at (
   1.73205
 ,
    .00000
 ){ 2};
 \draw[fill=
 gray 
 ](
  1.73205081,  2.
 )--(
  .866025404,  2.5
 )--(
  0.,  2.
 )--(
  0.,  1.
 )--(
  .866025404,  .5
 )--(
  1.73205081,  1.
 )--(
  1.73205081,  2.
 );
 \node at (
    .86603
 ,
   1.50000
 ){ 3};
 \draw[fill=
 yellow 
 ](
  0.,  2.
 )--(
 -.866025404,  2.5
 )--(
 -1.73205081,  2.
 )--(
 -1.73205081,  1.
 )--(
 -.866025404,  .5
 )--(
  0.,  1.
 )--(
  0.,  2.
 );
 \node at (
   -.86603
 ,
   1.50000
 ){ 4};
 \draw[fill=
 gray 
 ](
 -.866025404,  .5
 )--(
 -1.73205081,  1.
 )--(
 -2.59807621,  .5
 )--(
 -2.59807621, -.5
 )--(
 -1.73205081, -1.
 )--(
 -.866025404, -.5
 )--(
 -.866025404,  .5
 );
 \node at (
  -1.73205
 ,
    .00000
 ){ 5};
 \draw[fill=
 yellow 
 ](
  0., -1.
 )--(
 -.866025404, -.5
 )--(
 -1.73205081, -1.
 )--(
 -1.73205081, -2.
 )--(
 -.866025404, -2.5
 )--(
  0., -2.
 )--(
  0., -1.
 );
 \node at (
   -.86603
 ,
  -1.50000
 ){ 6};
 \draw[fill=
 white 
 ](
  1.73205081, -1.
 )--(
  .866025404, -.5
 )--(
  0., -1.
 )--(
  0., -2.
 )--(
  .866025404, -2.5
 )--(
  1.73205081, -2.
 )--(
  1.73205081, -1.
 );
 \node at (
    .86603
 ,
  -1.50000
 ){ 7};
 \draw[fill=
 yellow 
 ](
  4.33012702,  .5
 )--(
  3.46410162,  1.
 )--(
  2.59807621,  .5
 )--(
  2.59807621, -.5
 )--(
  3.46410162, -1.
 )--(
  4.33012702, -.5
 )--(
  4.33012702,  .5
 );
 \node at (
   3.46410
 ,
    .00000
 ){ 8};
 \draw[fill=
 white 
 ](
  3.46410162,  2.
 )--(
  2.59807621,  2.5
 )--(
  1.73205081,  2.
 )--(
  1.73205081,  1.
 )--(
  2.59807621,  .5
 )--(
  3.46410162,  1.
 )--(
  3.46410162,  2.
 );
 \node at (
   2.59808
 ,
   1.50000
 ){ 9};
 \draw[fill=
 yellow 
 ](
  2.59807621,  3.5
 )--(
  1.73205081,  4.
 )--(
  .866025404,  3.5
 )--(
  .866025404,  2.5
 )--(
  1.73205081,  2.
 )--(
  2.59807621,  2.5
 )--(
  2.59807621,  3.5
 );
 \node at (
   1.73205
 ,
   3.00000
 ){ 10};
 \draw[fill=
 gray 
 ](
  .866025404,  3.5
 )--(
  0.,  4.
 )--(
 -.866025404,  3.5
 )--(
 -.866025404,  2.5
 )--(
  0.,  2.
 )--(
  .866025404,  2.5
 )--(
  .866025404,  3.5
 );
 \node at (
    .00000
 ,
   3.00000
 ){ 11};
 \draw[fill=
 white 
 ](
 -.866025404,  3.5
 )--(
 -1.73205081,  4.
 )--(
 -2.59807621,  3.5
 )--(
 -2.59807621,  2.5
 )--(
 -1.73205081,  2.
 )--(
 -.866025404,  2.5
 )--(
 -.866025404,  3.5
 );
 \node at (
  -1.73205
 ,
   3.00000
 ){ 12};
 \draw[fill=
 white 
 ](
 -1.73205081,  2.
 )--(
 -2.59807621,  2.5
 )--(
 -3.46410162,  2.
 )--(
 -3.46410162,  1.
 )--(
 -2.59807621,  .5
 )--(
 -1.73205081,  1.
 )--(
 -1.73205081,  2.
 );
 \node at (
  -2.59808
 ,
   1.50000
 ){ 13};
 \draw[fill=
 black 
 ](
 -2.59807621,  .5
 )--(
 -3.46410162,  1.
 )--(
 -4.33012702,  .5
 )--(
 -4.33012702, -.5
 )--(
 -3.46410162, -1.
 )--(
 -2.59807621, -.5
 )--(
 -2.59807621,  .5
 );
 \node at (
  -3.46410
 ,
    .00000
 ){ 14};
 \draw[fill=
 white   
 ](
 -1.73205081, -1.
 )--(
 -2.59807621, -.5
 )--(
 -3.46410162, -1.
 )--(
 -3.46410162, -2.
 )--(
 -2.59807621, -2.5
 )--(
 -1.73205081, -2.
 )--(
 -1.73205081, -1.
 );
 \node at (
  -2.59808
 ,
  -1.50000
 ){ 15};
 \draw[fill=
 white 
 ](
 -.866025404, -2.5
 )--(
 -1.73205081, -2.
 )--(
 -2.59807621, -2.5
 )--(
 -2.59807621, -3.5
 )--(
 -1.73205081, -4.
 )--(
 -.866025404, -3.5
 )--(
 -.866025404, -2.5
 );
 \node at (
  -1.73205
 ,
  -3.00000
 ){ 16};
 \draw[fill=
 yellow   
 ](
  .866025404, -2.5
 )--(
  0., -2.
 )--(
 -.866025404, -2.5
 )--(
 -.866025404, -3.5
 )--(
  0., -4.
 )--(
  .866025404, -3.5
 )--(
  .866025404, -2.5
 );
 \node at (
    .00000
 ,
  -3.00000
 ){ 17};
 \draw[fill=
 yellow 
 ](
  2.59807621, -2.5
 )--(
  1.73205081, -2.
 )--(
  .866025404, -2.5
 )--(
  .866025404, -3.5
 )--(
  1.73205081, -4.
 )--(
  2.59807621, -3.5
 )--(
  2.59807621, -2.5
 );
 \node at (
   1.73205
 ,
  -3.00000
 ){ 18};
 \draw[fill=
 yellow   
 ](
  3.46410162, -1.
 )--(
  2.59807621, -.5
 )--(
  1.73205081, -1.
 )--(
  1.73205081, -2.
 )--(
  2.59807621, -2.5
 )--(
  3.46410162, -2.
 )--(
  3.46410162, -1.
 );
 \node at (
   2.59808
 ,
  -1.50000
 ){ 19};
 \draw[fill=
 white  
 ](
  6.06217783,  .5
 )--(
  5.19615242,  1.
 )--(
  4.33012702,  .5
 )--(
  4.33012702, -.5
 )--(
  5.19615242, -1.
 )--(
  6.06217783, -.5
 )--(
  6.06217783,  .5
 );
 \node at (
   5.19615
 ,
    .00000
 ){ 20};
 \draw[fill=
 yellow 
 ](
  5.19615242,  2.
 )--(
  4.33012702,  2.5
 )--(
  3.46410162,  2.
 )--(
  3.46410162,  1.
 )--(
  4.33012702,  .5
 )--(
  5.19615242,  1.
 )--(
  5.19615242,  2.
 );
 \node at (
   4.33013
 ,
   1.50000
 ){ 21};
 \draw[fill=
 yellow 
 ](
  4.33012702,  3.5
 )--(
  3.46410162,  4.
 )--(
  2.59807621,  3.5
 )--(
  2.59807621,  2.5
 )--(
  3.46410162,  2.
 )--(
  4.33012702,  2.5
 )--(
  4.33012702,  3.5
 );
 \node at (
   3.46410
 ,
   3.00000
 ){ 22};
 \draw[fill=
 white   
 ](
  3.46410162,  5.
 )--(
  2.59807621,  5.5
 )--(
  1.73205081,  5.
 )--(
  1.73205081,  4.
 )--(
  2.59807621,  3.5
 )--(
  3.46410162,  4.
 )--(
  3.46410162,  5.
 );
 \node at (
   2.59808
 ,
   4.50000
 ){ 23};
 \draw[fill=
 white 
 ](
  1.73205081,  5.
 )--(
  .866025404,  5.5
 )--(
  0.,  5.
 )--(
  0.,  4.
 )--(
  .866025404,  3.5
 )--(
  1.73205081,  4.
 )--(
  1.73205081,  5.
 );
 \node at (
    .86603
 ,
   4.50000
 ){ 24};
 \draw[fill=
 black   
 ](
  0.,  5.
 )--(
 -.866025404,  5.5
 )--(
 -1.73205081,  5.
 )--(
 -1.73205081,  4.
 )--(
 -.866025404,  3.5
 )--(
  0.,  4.
 )--(
  0.,  5.
 );
 \node at (
   -.86603
 ,
   4.50000
 ){ 25};
 \draw[fill=
 gray 
 ](
 -1.73205081,  5.
 )--(
 -2.59807621,  5.5
 )--(
 -3.46410162,  5.
 )--(
 -3.46410162,  4.
 )--(
 -2.59807621,  3.5
 )--(
 -1.73205081,  4.
 )--(
 -1.73205081,  5.
 );
 \node at (
  -2.59808
 ,
   4.50000
 ){ 26};
 \draw[fill=
 black   
 ](
 -2.59807621,  3.5
 )--(
 -3.46410162,  4.
 )--(
 -4.33012702,  3.5
 )--(
 -4.33012702,  2.5
 )--(
 -3.46410162,  2.
 )--(
 -2.59807621,  2.5
 )--(
 -2.59807621,  3.5
 );
 \node at (
  -3.46410
 ,
   3.00000
 ){ 27};
 \draw[fill=
 gray 
 ](
 -3.46410162,  2.
 )--(
 -4.33012702,  2.5
 )--(
 -5.19615242,  2.
 )--(
 -5.19615242,  1.
 )--(
 -4.33012702,  .5
 )--(
 -3.46410162,  1.
 )--(
 -3.46410162,  2.
 );
 \node at (
  -4.33013
 ,
   1.50000
 ){ 28};
 \draw[fill=
 white   
 ](
 -4.33012702,  .5
 )--(
 -5.19615242,  1.
 )--(
 -6.06217783,  .5
 )--(
 -6.06217783, -.5
 )--(
 -5.19615242, -1.
 )--(
 -4.33012702, -.5
 )--(
 -4.33012702,  .5
 );
 \node at (
  -5.19615
 ,
    .00000
 ){ 29};
 \draw[fill=
 gray 
 ](
 -3.46410162, -1.
 )--(
 -4.33012702, -.5
 )--(
 -5.19615242, -1.
 )--(
 -5.19615242, -2.
 )--(
 -4.33012702, -2.5
 )--(
 -3.46410162, -2.
 )--(
 -3.46410162, -1.
 );
 \node at (
  -4.33013
 ,
  -1.50000
 ){ 30};
 \end{tikzpicture}
 \label{fig:greyblackyellow}}\qquad \qquad 
\subfigure[An Extreme Point with coefficients $0,\frac{1}{3}, \frac{2}{3}.$]{
 \begin{tikzpicture}[scale=0.5]
 \tikzstyle{every node}=[font=\tiny]
 \draw[fill=
 green 
 ](
  .866025404,  .5
 )--(
  0.,  1.
 )--(
 -.866025404,  .5
 )--(
 -.866025404, -.5
 )--(
  0., -1.
 )--(
  .866025404, -.5
 )--(
  .866025404,  .5
 );
 \node at (
    .00000
 ,
    .00000
 ){ 1};
 \draw[fill=
 white 
 ](
  2.59807621,  .5
 )--(
  1.73205081,  1.
 )--(
  .866025404,  .5
 )--(
  .866025404, -.5
 )--(
  1.73205081, -1.
 )--(
  2.59807621, -.5
 )--(
  2.59807621,  .5
 );
 \node at (
   1.73205
 ,
    .00000
 ){ 2};
 \draw[fill=
 green 
 ](
  1.73205081,  2.
 )--(
  .866025404,  2.5
 )--(
  0.,  2.
 )--(
  0.,  1.
 )--(
  .866025404,  .5
 )--(
  1.73205081,  1.
 )--(
  1.73205081,  2.
 );
 \node at (
    .86603
 ,
   1.50000
 ){ 3};
 \draw[fill=
 green 
 ](
  0.,  2.
 )--(
 -.866025404,  2.5
 )--(
 -1.73205081,  2.
 )--(
 -1.73205081,  1.
 )--(
 -.866025404,  .5
 )--(
  0.,  1.
 )--(
  0.,  2.
 );
 \node at (
   -.86603
 ,
   1.50000
 ){ 4};
 \draw[fill=
 green 
 ](
 -.866025404,  .5
 )--(
 -1.73205081,  1.
 )--(
 -2.59807621,  .5
 )--(
 -2.59807621, -.5
 )--(
 -1.73205081, -1.
 )--(
 -.866025404, -.5
 )--(
 -.866025404,  .5
 );
 \node at (
  -1.73205
 ,
    .00000
 ){ 5};
 \draw[fill=
 green 
 ](
  0., -1.
 )--(
 -.866025404, -.5
 )--(
 -1.73205081, -1.
 )--(
 -1.73205081, -2.
 )--(
 -.866025404, -2.5
 )--(
  0., -2.
 )--(
  0., -1.
 );
 \node at (
   -.86603
 ,
  -1.50000
 ){ 6};
 \draw[fill=
 white 
 ](
  1.73205081, -1.
 )--(
  .866025404, -.5
 )--(
  0., -1.
 )--(
  0., -2.
 )--(
  .866025404, -2.5
 )--(
  1.73205081, -2.
 )--(
  1.73205081, -1.
 );
 \node at (
    .86603
 ,
  -1.50000
 ){ 7};
 \draw[fill=
 green 
 ](
  4.33012702,  .5
 )--(
  3.46410162,  1.
 )--(
  2.59807621,  .5
 )--(
  2.59807621, -.5
 )--(
  3.46410162, -1.
 )--(
  4.33012702, -.5
 )--(
  4.33012702,  .5
 );
 \node at (
   3.46410
 ,
    .00000
 ){ 8};
 \draw[fill=
 green 
 ](
  3.46410162,  2.
 )--(
  2.59807621,  2.5
 )--(
  1.73205081,  2.
 )--(
  1.73205081,  1.
 )--(
  2.59807621,  .5
 )--(
  3.46410162,  1.
 )--(
  3.46410162,  2.
 );
 \node at (
   2.59808
 ,
   1.50000
 ){ 9};
 \draw[fill=
 green 
 ](
  2.59807621,  3.5
 )--(
  1.73205081,  4.
 )--(
  .866025404,  3.5
 )--(
  .866025404,  2.5
 )--(
  1.73205081,  2.
 )--(
  2.59807621,  2.5
 )--(
  2.59807621,  3.5
 );
 \node at (
   1.73205
 ,
   3.00000
 ){ 10};
 \draw[fill=
 green 
 ](
  .866025404,  3.5
 )--(
  0.,  4.
 )--(
 -.866025404,  3.5
 )--(
 -.866025404,  2.5
 )--(
  0.,  2.
 )--(
  .866025404,  2.5
 )--(
  .866025404,  3.5
 );
 \node at (
    .00000
 ,
   3.00000
 ){ 11};
 \draw[fill=
 white 
 ](
 -.866025404,  3.5
 )--(
 -1.73205081,  4.
 )--(
 -2.59807621,  3.5
 )--(
 -2.59807621,  2.5
 )--(
 -1.73205081,  2.
 )--(
 -.866025404,  2.5
 )--(
 -.866025404,  3.5
 );
 \node at (
  -1.73205
 ,
   3.00000
 ){ 12};
 \draw[fill=
 white 
 ](
 -1.73205081,  2.
 )--(
 -2.59807621,  2.5
 )--(
 -3.46410162,  2.
 )--(
 -3.46410162,  1.
 )--(
 -2.59807621,  .5
 )--(
 -1.73205081,  1.
 )--(
 -1.73205081,  2.
 );
 \node at (
  -2.59808
 ,
   1.50000
 ){ 13};
 \draw[fill=
 red
 ](
 -2.59807621,  .5
 )--(
 -3.46410162,  1.
 )--(
 -4.33012702,  .5
 )--(
 -4.33012702, -.5
 )--(
 -3.46410162, -1.
 )--(
 -2.59807621, -.5
 )--(
 -2.59807621,  .5
 );
 \node at (
  -3.46410
 ,
    .00000
 ){ 14};
 \draw[fill=
 white   
 ](
 -1.73205081, -1.
 )--(
 -2.59807621, -.5
 )--(
 -3.46410162, -1.
 )--(
 -3.46410162, -2.
 )--(
 -2.59807621, -2.5
 )--(
 -1.73205081, -2.
 )--(
 -1.73205081, -1.
 );
 \node at (
  -2.59808
 ,
  -1.50000
 ){ 15};
 \draw[fill=
 white 
 ](
 -.866025404, -2.5
 )--(
 -1.73205081, -2.
 )--(
 -2.59807621, -2.5
 )--(
 -2.59807621, -3.5
 )--(
 -1.73205081, -4.
 )--(
 -.866025404, -3.5
 )--(
 -.866025404, -2.5
 );
 \node at (
  -1.73205
 ,
  -3.00000
 ){ 16};
 \draw[fill=
 red   
 ](
  .866025404, -2.5
 )--(
  0., -2.
 )--(
 -.866025404, -2.5
 )--(
 -.866025404, -3.5
 )--(
  0., -4.
 )--(
  .866025404, -3.5
 )--(
  .866025404, -2.5
 );
 \node at (
    .00000
 ,
  -3.00000
 ){ 17};
 \draw[fill=
 green 
 ](
  2.59807621, -2.5
 )--(
  1.73205081, -2.
 )--(
  .866025404, -2.5
 )--(
  .866025404, -3.5
 )--(
  1.73205081, -4.
 )--(
  2.59807621, -3.5
 )--(
  2.59807621, -2.5
 );
 \node at (
   1.73205
 ,
  -3.00000
 ){ 18};
 \draw[fill=
 red   
 ](
  3.46410162, -1.
 )--(
  2.59807621, -.5
 )--(
  1.73205081, -1.
 )--(
  1.73205081, -2.
 )--(
  2.59807621, -2.5
 )--(
  3.46410162, -2.
 )--(
  3.46410162, -1.
 );
 \node at (
   2.59808
 ,
  -1.50000
 ){ 19};
 \draw[fill=
 green  
 ](
  6.06217783,  .5
 )--(
  5.19615242,  1.
 )--(
  4.33012702,  .5
 )--(
  4.33012702, -.5
 )--(
  5.19615242, -1.
 )--(
  6.06217783, -.5
 )--(
  6.06217783,  .5
 );
 \node at (
   5.19615
 ,
    .00000
 ){ 20};
 \draw[fill=
 green 
 ](
  5.19615242,  2.
 )--(
  4.33012702,  2.5
 )--(
  3.46410162,  2.
 )--(
  3.46410162,  1.
 )--(
  4.33012702,  .5
 )--(
  5.19615242,  1.
 )--(
  5.19615242,  2.
 );
 \node at (
   4.33013
 ,
   1.50000
 ){ 21};
 \draw[fill=
 green 
 ](
  4.33012702,  3.5
 )--(
  3.46410162,  4.
 )--(
  2.59807621,  3.5
 )--(
  2.59807621,  2.5
 )--(
  3.46410162,  2.
 )--(
  4.33012702,  2.5
 )--(
  4.33012702,  3.5
 );
 \node at (
   3.46410
 ,
   3.00000
 ){ 22};
 \draw[fill=
 green   
 ](
  3.46410162,  5.
 )--(
  2.59807621,  5.5
 )--(
  1.73205081,  5.
 )--(
  1.73205081,  4.
 )--(
  2.59807621,  3.5
 )--(
  3.46410162,  4.
 )--(
  3.46410162,  5.
 );
 \node at (
   2.59808
 ,
   4.50000
 ){ 23};
 \draw[fill=
 white 
 ](
  1.73205081,  5.
 )--(
  .866025404,  5.5
 )--(
  0.,  5.
 )--(
  0.,  4.
 )--(
  .866025404,  3.5
 )--(
  1.73205081,  4.
 )--(
  1.73205081,  5.
 );
 \node at (
    .86603
 ,
   4.50000
 ){ 24};
 \draw[fill=
 red  
 ](
  0.,  5.
 )--(
 -.866025404,  5.5
 )--(
 -1.73205081,  5.
 )--(
 -1.73205081,  4.
 )--(
 -.866025404,  3.5
 )--(
  0.,  4.
 )--(
  0.,  5.
 );
 \node at (
   -.86603
 ,
   4.50000
 ){ 25};
 \draw[fill=
 green 
 ](
 -1.73205081,  5.
 )--(
 -2.59807621,  5.5
 )--(
 -3.46410162,  5.
 )--(
 -3.46410162,  4.
 )--(
 -2.59807621,  3.5
 )--(
 -1.73205081,  4.
 )--(
 -1.73205081,  5.
 );
 \node at (
  -2.59808
 ,
   4.50000
 ){ 26};
 \draw[fill=
 red   
 ](
 -2.59807621,  3.5
 )--(
 -3.46410162,  4.
 )--(
 -4.33012702,  3.5
 )--(
 -4.33012702,  2.5
 )--(
 -3.46410162,  2.
 )--(
 -2.59807621,  2.5
 )--(
 -2.59807621,  3.5
 );
 \node at (
  -3.46410
 ,
   3.00000
 ){ 27};
 \draw[fill=
 green
 ](
 -3.46410162,  2.
 )--(
 -4.33012702,  2.5
 )--(
 -5.19615242,  2.
 )--(
 -5.19615242,  1.
 )--(
 -4.33012702,  .5
 )--(
 -3.46410162,  1.
 )--(
 -3.46410162,  2.
 );
 \node at (
  -4.33013
 ,
   1.50000
 ){ 28};
 \draw[fill=
 white   
 ](
 -4.33012702,  .5
 )--(
 -5.19615242,  1.
 )--(
 -6.06217783,  .5
 )--(
 -6.06217783, -.5
 )--(
 -5.19615242, -1.
 )--(
 -4.33012702, -.5
 )--(
 -4.33012702,  .5
 );
 \node at (
  -5.19615
 ,
    .00000
 ){ 29};
 \draw[fill=
 green 
 ](
 -3.46410162, -1.
 )--(
 -4.33012702, -.5
 )--(
 -5.19615242, -1.
 )--(
 -5.19615242, -2.
 )--(
 -4.33012702, -2.5
 )--(
 -3.46410162, -2.
 )--(
 -3.46410162, -1.
 );
 \node at (
  -4.33013
 ,
  -1.50000
 ){ 30};
 \end{tikzpicture}
  \label{fig:whitegreenred}
}
\caption{Exotic Extreme Points}
\end{figure}

%%% Local Variables: 
%%% mode: latex
%%% TeX-master: "graphs"
%%% End: 